\documentclass[conference]{IEEEtran}
\IEEEoverridecommandlockouts
\usepackage{cite}
\usepackage{amsmath,amssymb,amsfonts}
\usepackage{algorithmic}
\usepackage{algorithm}
\usepackage{graphicx}
\usepackage{textcomp}
\usepackage{xcolor}
\usepackage{stfloats}
\usepackage{tikz}

\usepackage[left=0.67in,right=0.67in,top=0.7in,bottom=1.1in]{geometry}

\usepackage{lipsum}
\usepackage{mathtools}
\usepackage{cuted}

\usepackage[nodisplayskipstretch]{setspace}
\usepackage{bm}
\usepackage{subfigure}
\usepackage{amsthm} 
\usepackage{url}

\def\BibTeX{{\rm B\kern-.05em{\sc i\kern-.025em b}\kern-.08em
		T\kern-.1667em\lower.7ex\hbox{E}\kern-.125emX}}
\newcommand{\mr}[1]{\mathrm{#1}}

\newtheorem{proposition}{Proposition}
\newtheorem{remark}{Remark}

\newcommand{\br}[1]{\bm{\mathrm{#1}}} 

\makeatletter
\def\endthebibliography{%
	\def\@noitemerr{\@latex@warning{Empty `thebibliography' environment}}%
	\endlist
}
\makeatother

\begin{document}
	\title{Rotating ULA-Enabled Computed Tomography for Efficient 3D Spatial Power Spectrum Synthesis: Architecture and Principled Orientation Design \\
		} 
		\author{Haocheng Hua, \textit{Member, IEEE}, Weidong Mei, \textit{Member, IEEE}, Jie Xu, \textit{Fellow, IEEE}, Rui Zhang, \textit{Fellow, IEEE} \\
					\thanks{Part of this paper has been presented at IEEE Vehicular Technology Conference (VTC), Nice, France, June 9-12, 2026 \cite{hua2026rotating}.}
			\thanks{
				Haocheng Hua is with the School of Science and Engineering, The Chinese University of Hong Kong (Shenzhen), Guangdong 518172, China (e-mail: huahaocheng@cuhk.edu.cn).}
			\thanks{Weidong Mei is with the National Key Laboratory of Wireless Communications, University of Electronic Science and Technology of China, Chengdu 611731, China (e-mail: wmei@uestc.edu.cn).} 
			\thanks{Jie Xu is with the School of Science and Engineering, the Shenzhen Future Network of Intelligence Institute (FNii-Shenzhen), and the Guangdong Provincial Key Laboratory of Future Networks of Intelligence, The Chinese University of Hong Kong (Shenzhen), Guangdong 518172, China (email: xujie@cuhk.edu.cn).}
			\thanks{Rui Zhang is with the Department of Electrical and Computer Engineering, National University of Singapore, Singapore 117583 (e-mail: elezhang@nus.edu.sg). R. Zhang is the corresponding author.} 			 
		}
		
		\maketitle
		\begin{abstract}
			This paper proposes an efficient three-dimensional (3D) spatial power spectrum synthesis method by rotating a uniform linear array (ULA) about its center in 3D space. 
			Inspired by classical computed tomography (CT), the ULA performs analog receive combining at each rotation angle to produce a partial coherent sum. 
			By collecting such sums over multiple rotations, the full 3D spectrum can be synthesized online via a single radio-frequency (RF) chain, without explicitly acquiring per-antenna signals. 
			Depending on whether the overall coherent sum is accessible, the synthesis is obtained through either a minimum operation over partial spectrum images or joint synthesis after accumulating all coherent sums. 
			Compared with fixed uniform planar array (UPA)-based combining and dense single movable-antenna (MA) sampling for 3D cubic virtual arrays, the proposed scheme achieves full-space 3D coverage with substantially reduced sampling and movement overhead while maintaining uniformly high angular resolution. 
			Its sampling geometry and sequential orientation design also support pipelined analog beamforming, reducing practical hardware cost.  
			To design rotation orientations in a principled manner without prior environmental information, we aim to maximize the expected worst-case projected separation between multi-path component (MPC) pairs. 
			We show that the globally optimal relaxed design is isotropic in the sense that the summed outer products of all orientation axes form a scaled identity matrix, and this design can be realized by finite orientation sets for any number of orientations no smaller than three. 
			A secondary criterion then minimizes the worst-case projective correlation among orientation axes to reduce orientation redundancy. 
			Accordingly, we optimize orientation sets for different numbers of orientations under isotropic-matrix and unit-norm constraints, using a multi-start smooth minimax algorithm. 
			Numerical results show that the optimized orientations uniformly span 3D space and reconstruct full-space 3D spectra close to the dense 3D cubic reference using only a fraction of spatial samples.
		\end{abstract}

		\section{Introduction} 
		\label{sec:intro}

		Multi-antenna technology has been a cornerstone of fourth-generation (4G) and fifth-generation (5G) wireless networks \cite{heath2018foundations}. Looking ahead to next-generation/sixth-generation (6G) networks, emerging extremely large-scale antenna array (ELAA) \cite{lu2024tutorial,hua2025near}, intelligent reflecting surface (IRS)/reconfigurable intelligent surface (RIS) \cite{wu2024intelligent}, and movable antenna (MA)/six-dimensional movable antenna (6DMA) technologies \cite{zhu2025tutorial,shao20246d,hua2025hierarchically} are expected to provide new spatial degrees of freedom (DoFs) for both communication and wireless sensing. The exploitation of such DoFs and the design and configuration of those different emerging wireless systems, however, rely critically on timely and reliable characterization of the wireless environment. In particular, beyond estimating instantaneous channel coefficients at a finite set of antennas, many wireless applications could benefit from a compact description of how the received power is distributed over the three-dimensional (3D) angular domain.
		
		The spatial power spectrum (SPS), also referred to as the angular multi-path power profile, provides such a representation by depicting the received signal power distribution in the angular domain \cite{zhao2023nerf2,medbo201560,medbo2016frequency,zhou20256d}. Compared with location-dependent received signal strength (RSS) or a sparse list of estimated multi-path parameters \cite{zeng2024tutorial}, the SPS directly preserves the angular structure of dominant multi-path components (MPCs) caused by line-of-sight (LoS) propagation and specular reflections, as well as diffuse multi-path components (DMCs) induced by dense scattering \cite{medbo201560}. Moreover, by representing the wireless environment as a spectrum image, the SPS is naturally compatible with modern image-based processing and learning pipelines. This makes efficient SPS synthesis an important primitive for MA/6DMA array configuration \cite{zhu2025tutorial,shao20246d,hua2025hierarchically}, IRS/RIS phase tuning \cite{wu2024intelligent}, millimeter-wave beam alignment \cite{heath2016overview}, among others. Besides, efficient SPS synthesis pipeline could accelerate the SPS data acquisition and collection in different wireless environments, which is crucial for the wireless environment sensing and learning applications.
		
		Despite its importance, synthesizing a full-space 3D SPS with high angular resolution remains challenging. One conventional approach is to use a uniform planar array (UPA) with analog receive combining to sweep candidate azimuth/elevation directions and directly measure the beamformed output power, i.e., online synthesis \cite{zhao2023nerf2}. This approach and its underlying receiver architecture are efficient in the sense that they can operate with a single radio-frequency (RF) chain, but a planar array inherently suffers from front-back ambiguity and therefore provides only half-space coverage. In addition, its angular resolution becomes non-uniform and deteriorates near directions close to the array plane due to the shrinkage of the projected aperture. Another approach forms a dense 3D cubic virtual array by moving an antenna element over a volumetric grid and storing the measured received signals at each position. Then, the full-space 3D SPS is synthesized offline based on all the measured signals via beamforming \cite{medbo201560,medbo2016frequency}. This method can provide full-space coverage and nearly uniform angular resolution, but it requires explicitly measuring and storing the received signals over massive spatial positions for a cubic grid. Such dense volumetric sampling causes substantial acquisition latency and movement overhead, and it is generally impractical to build a gigantic 3D phased array to cover all the sampling positions and implement online analog combining.
		
		To address these issues, we propose in this paper an efficient rotating uniform linear array (RULA)-enabled computed tomography (CT) framework for online full-space 3D SPS synthesis. Our main results are summarized as follows:
		\begin{itemize}
			\item We first establish the definition of full-space 3D SPS, based on which we review two representative conventional approaches, namely online analog receive combining with a fixed UPA and offline dense sampling via a single MA to form a 3D cubic virtual array. This comparison shows that the SPS synthesized via the UPA-based approach is limited to half-space coverage with non-uniform resolution, whereas the 3D cubic virtual-array approach provides full-space high-resolution synthesis at the cost of massive spatial samples and substantial antenna movement overhead.
			\item We then propose the RULA-enabled online synthesis framework. By rotating a ULA about its center in 3D space, the receiver performs analog receive combining at each orientation to produce a partial coherent sum. By collecting such partial coherent sums over multiple rotations, the full 3D SPS can be synthesized online via a single RF chain. We then develop two synthesis rules under the same RULA sampling geometry. When only partial-spectrum images are available, the full-space SPS is obtained through a pointwise minimum fusion operation. When the overall coherent sum across rotations is accessible, we further propose a joint coherent fusion operation by accumulating all collected coherent sums. The implementation feasibility and requirements of 3D cubic sampling, UPA sampling, and RULA sampling are then compared with each other, further highlighting the efficiency of the proposed RULA-CT framework.
			\item To facilitate principled orientation design, we further propose two design criteria for orientation design without requiring prior information about the surrounding environment. The primary criterion maximizes the expected worst-case projected separation between randomly distributed MPC pairs. We show that the globally optimal relaxed design is isotropic in the sense that the summed outer products of all orientation axes form a scaled identity matrix, and prove that this condition can be exactly realized by finite unit-norm tight frame orientation sets. We also derive the scaling law of the corresponding objective value. After the primary criterion is satisfied, we further propose a secondary criterion that minimizes the worst-case projective correlation among the ULA orientation axes. Based on the two criteria, we formulate an optimization problem that minimizes the worst-case projective correlation subject to the isotropic-matrix and unit-norm constraints. A multi-start smooth minimax algorithm is then proposed to optimize orientation designs for different numbers of orientations.
			\item Numerical results validate the efficacy of our proposed orientation design and SPS synthesis framework. The optimized orientation sets are shown to span the 3D space uniformly. With the optimized design, the proposed RULA-CT method reconstructs full-space 3D spectra close to the dense 3D cubic reference using only a fraction of the spatial samples required by dense 3D cubic sampling. The results also show that joint coherent fusion is more robust in the low signal-to-noise ratio (SNR) regime due to coherent aggregation, while pointwise minimum fusion achieves higher structural fidelity in the high-SNR regime by suppressing sidelobe and ambiguity artifacts.
		\end{itemize}
		
		The remainder of this paper is organized as follows. Section II presents the channel model, the definition of SPS, and reviews conventional full-space offline and half-space online SPS synthesis methods. Section III introduces the proposed RULA-enabled online synthesis framework and compares it with existing sampling strategies. Section IV develops the environment-agnostic orientation design, followed by the formulated optimization problem and the proposed multi-start smooth minimax algorithm. Section V provides numerical results to validate the proposed design and compare it with benchmark methods. Section VI concludes this paper.
		
		\textit{Notations:} Boldface lower-case and upper-case letters denote vectors and matrices, respectively. For a matrix \(\br M\), \(\br M^H\), \(\br M^T\), \(\br M^{-1}\), \(\det(\br M)\), \(\operatorname{tr}(\br M)\), \(\|\br M\|_F\), \([\br M]_{i,j}\), and \(\br M[:,i]\) denote its conjugate transpose, transpose, inverse, determinant, trace, Frobenius norm, \((i,j)\)-th entry, and \(i\)-th column, respectively. \(\br M\succeq 0\) means that \(\br M\) is positive semidefinite. \(\mathbb R^{x\times y}\) and \(\mathbb C^{x\times y}\) denote the spaces of \(x\times y\) real and complex matrices, respectively. \(\mathbb E\{\cdot\}\) denotes statistical expectation, and \(\Pr(\cdot)\) denotes probability. \(\|\br x\|\) denotes the Euclidean norm of vector \(\br x\), \(|x|\) denotes the magnitude of scalar \(x\), and \(\br I_n\) denotes the \(n\times n\) identity matrix. \(\mr j=\sqrt{-1}\) is the imaginary unit. The distribution of a circularly symmetric complex Gaussian (CSCG) random vector with mean \(\br 0\) and covariance matrix \(\br \Sigma\) is denoted by \(\mathcal{CN}(\br 0,\br \Sigma)\), \(\mathcal N(\br 0,\br \Sigma)\) denotes the real Gaussian distribution, \(\mathcal U(a,b)\) denotes the uniform distribution over \([a,b]\), \(\operatorname{Unif}(\mathcal S^2)\) denotes the uniform distribution on the unit sphere \(\mathcal S^2\), and \(\sim\) stands for ``distributed as''. \(\min\{a,b\}\) and \(\max\{a,b\}\) return the minimum and maximum of \(a\) and \(b\), respectively. \(\lfloor\cdot\rfloor\) denotes the flooring operation.
		
		\section{Channel Model \& Spatial Spectrum Synthesis}
		
		Consider a narrowband far-field continuous angular channel model.
		Let $\br{p} \in \mathbb{R}^3$ denote a position within a given region $\mathcal{C}$, and let $\br{\nu} \in \Omega$
		be a unit direction vector on the unit sphere $\Omega$.
		The received complex field at position $\br{p}$ can be written as
		\begin{align}\label{eq:cont_r_p}
			r(\mathbf p)=\int_{\Omega} a(\boldsymbol{\nu})\,
			e^{j\frac{2\pi}{\lambda}\boldsymbol{\nu}^{T}\mathbf p}\, d\Omega,
		\end{align}
		where $a(\br{\nu})$ is the 
		complex channel gain at direction $\br{\nu}$, $\lambda$ is the carrier wavelength, and $\br{\nu}$ is given by
		\begin{align}\label{eq:nu_ori_def}
			\br{\nu} = 
			[u, v, w]^T =
			[\cos \theta \cos \phi, \cos \theta \sin \phi, \sin \theta ]^T,
		\end{align}
		with $\phi \in \left[-\pi,\pi\right]$ and $\theta \in \left[-\pi/2,\pi/2\right]$.
		The continuous
		spatial Fourier transform of the field $r(\br{p})$ with respect to direction $\br{\nu}$ is defined as
		\begin{align}\label{eq:cont_SPS}
			R(\boldsymbol{\nu})
			=
			\int_{\mathcal{C}}
			r(\mathbf p)\,
			e^{-j\frac{2\pi}{\lambda}\boldsymbol{\nu}^{T}\mathbf p}\,
			d\mathbf p .
		\end{align}
		The continuous 3D spatial power spectrum is then  defined as
		\begin{align}
			P(\boldsymbol{\nu})	=
			\left|R(\boldsymbol{\nu})\right|^{2}.
		\end{align}
		
		\begin{figure}[t]
			\centering
			\includegraphics[width=1.7in]{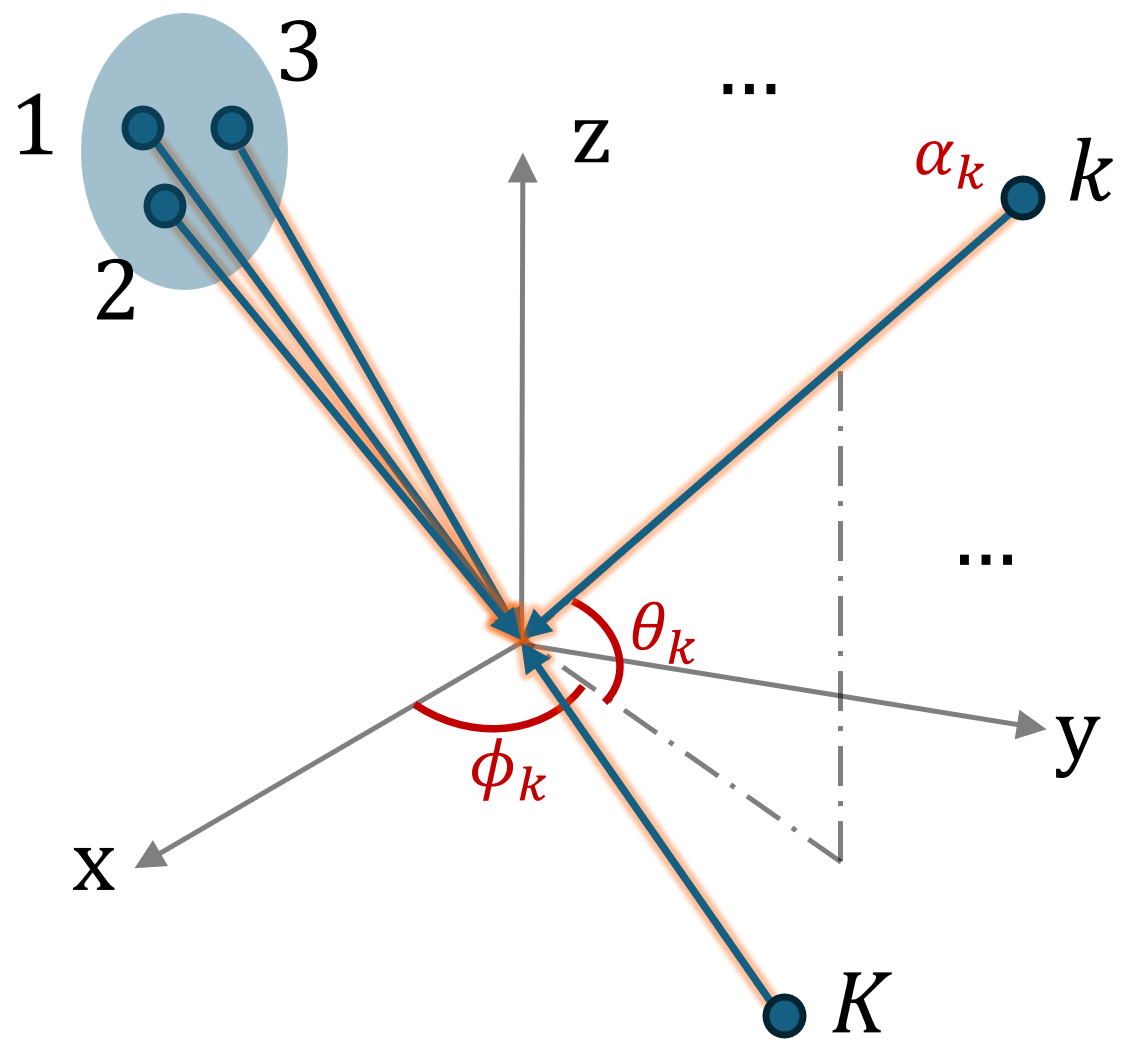}
			\centering
			\caption{Considered discrete channel model with $K$ MPCs.}
			\label{fig:plane_wave_def}
		\end{figure}
		To facilitate performance evaluation and analysis, as shown in Fig. \ref{fig:plane_wave_def}, we adopt the following discrete $K$-path channel model: 
		\begin{align}\label{eq:discrete_K_model}
			a(\boldsymbol{\nu})
			=\sum_{k=1}^{K} \alpha_k\,\delta_{\Omega}\!\left(\boldsymbol{\nu}-\boldsymbol{\nu}_k\right),
		\end{align}
		where $\alpha_k \in \mathbb{C}$ is the complex gain of the $k$-th path and $\br{\nu}_k \in \Omega$ is its angle-of-arrival (AoA) direction.
		Under the considered narrowband far-field model,
		substituting (\ref{eq:discrete_K_model}) into (\ref{eq:cont_r_p}) and using the sifting property of the Dirac delta function defined on $\Omega$ given by $
		\int_{\Omega} x(\br{\nu}) \delta _{\Omega} (\br{\nu}-\br{\nu}_k) d \Omega = x(\br{\nu}_k),$
		we have
		\begin{align}\label{eq:signal_model_K_discrete}
			r(\mathbf p)=
			\sum_{k=1}^{K}\alpha_k\,
			e^{j\frac{2\pi}{\lambda}\boldsymbol{\nu}_k^{T}\mathbf p}.
		\end{align}

Here, $r(\mathbf p)$ denotes the equivalent pilot-normalized complex baseband field that would be observed at position $\mathbf p$, 
with the known pilot symbol or known dedicated sensing signals absorbed into $\alpha_k$.
		In practice, the equivalent field sample $r(\mathbf p)$ is corrupted by receiver front-end noise. 
		Thus, we model the noisy sample associated with position $\br p$ as $\hat{r}(\br p)=r(\br p)+z(\br p)$,
		where $z(\br p)\sim\mathcal{CN}(0,\sigma^2)$ denotes the complex additive white Gaussian noise (AWGN) with zero mean and power $\sigma^2$, which is independent over $\br p$.
		
		\begin{figure}[t]
			\centering 
			\includegraphics[width=3.3in]{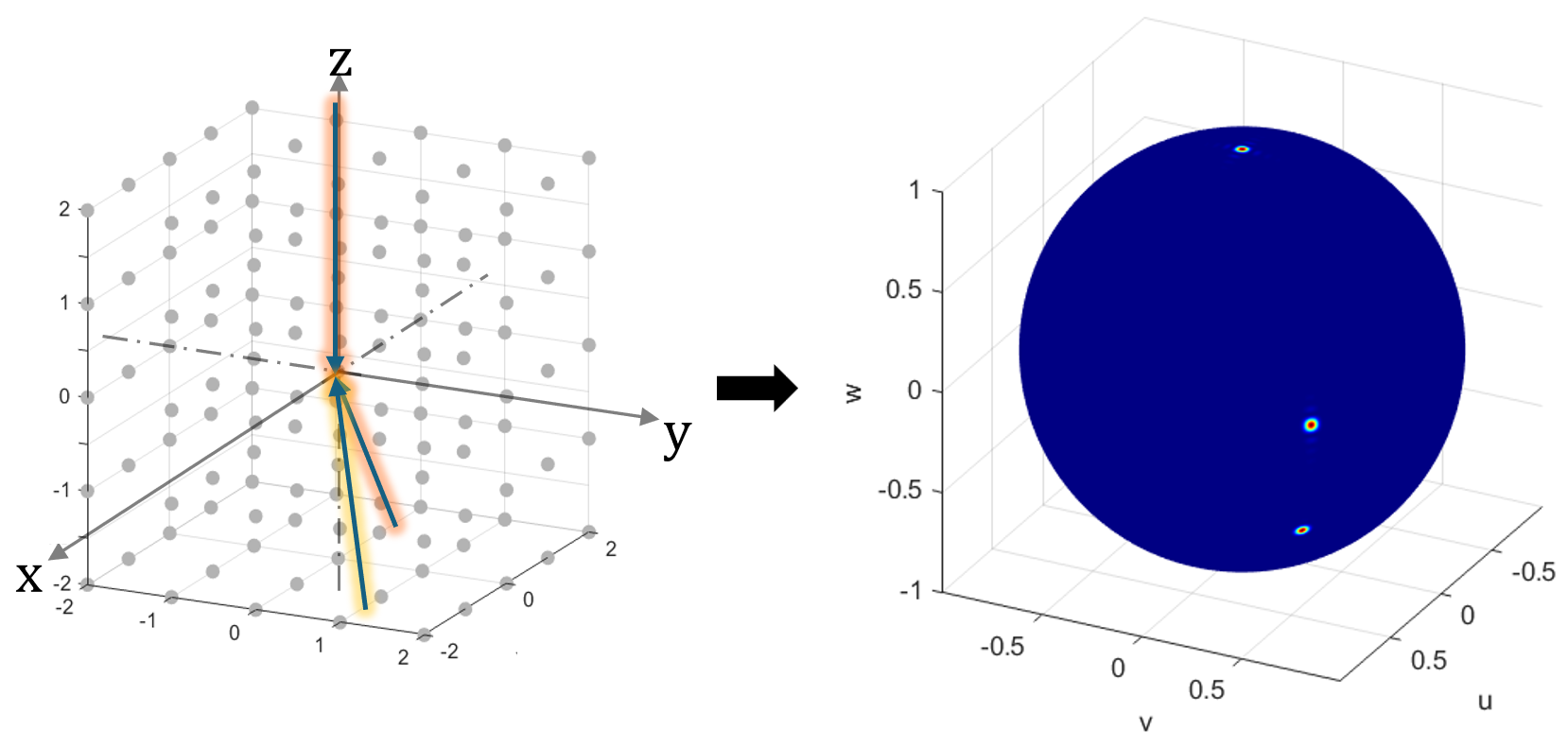}
			\centering
			\caption{An example of conventional full 3D spatial spectrum synthesized offline via forming virtual 3D cubic array with three MPCs at $(45^{\circ}, 5^{\circ}), (45^{\circ}, 90^{\circ})$, and $(45^{\circ}, -30^{\circ})$, respectively.}
			\label{fig:UPA_3D_full_SPP2}
		\end{figure} 
		
		\subsection{Conventional Full 3D Space Offline Synthesis}
		
		To evaluate (\ref{eq:cont_SPS}) over $\mathbb{R}^3$ in practice, it is conventional to adopt a finite 3D cubic region 
		$\mathcal{C} \subset \mathbb{R}^3$ whose center is located at the origin of the 3D Cartesian coordinate system, 
		and sample $r(\br{p})$ on a uniform $U \times U \times U$ grid\footnote{Without loss of generality, we assume $U$ to be an odd number for convenience.} \cite{medbo201560}, i.e., 
		\begin{align}\label{eq:samp_3D_cubic}
			&\br{p}_{n_x,n_y,n_z}  = \left[(n_x-l_{U})d,(n_y-l_{U})d,(n_z-l_{U})d\right]^T,\\
			\nonumber
			& n_x = 0,...,U-1, \enspace n_y = 0,...,U-1, \enspace n_z = 0,...,U-1,
		\end{align}
		where $l_U \triangleq \frac{U-1}{2}$,
		the total number of samples is $N = U^3$,
		and $d$ is the inter-element spacing.
		Then, the integral over $\mathcal{C}$ in (\ref{eq:cont_SPS}) can be approximated by a Riemann sum:
		\begin{align}
			\int_{\mathcal C}
			r(\mathbf p)\,e^{-j\frac{2\pi}{\lambda}\boldsymbol{\nu}^{T}\mathbf p}\,d\mathbf p
			\;\approx\;
			d^{3}\sum_{n=0}^{N-1}
			r(\mathbf p_n)\,e^{-j\frac{2\pi}{\lambda}\boldsymbol{\nu}^{T}\mathbf p_n}.
		\end{align}
		Accordingly, up to a scaling constant,
		the corresponding estimated discrete spatial power spectrum by considering the receiver front-end noise is given as  \cite{zhao2023nerf2}
		\begin{align}\label{eq:conventional_3D_SPSup}
			\mathbb{P}(\phi,\theta) \triangleq \widehat{P}(\boldsymbol{\nu}) =
			\frac{1}{N}
			\left|
			\sum_{n=0}^{N-1} 
			\hat{r}(\br{p}_n)
			e^{-j\frac{2\pi}{\lambda}\boldsymbol{\nu}^{T}\mathbf p_n}
			\right|^{2}.
		\end{align}
		Typically, the whole angular region $(\phi,\theta)$ is discretized into $Q_{\text{full}}$ bins. 
		Since physically building a gigantic 3D cubic phased array is generally impractical, the default approach is forming a virtual array by sequentially moving a single antenna element to each sampling position \(\br{p}_{n_x,n_y,n_z}\), explicitly measuring and storing the received signal samples \(\{\hat{r}(\br{p}_n)\}\), and then synthesizing the spatial spectrum offline based on (\ref{eq:conventional_3D_SPSup}).		
		As illustrated in Fig.~\ref{fig:UPA_3D_full_SPP2} and also demonstrated extensively via real channel environment measurements in \cite{medbo201560}, the resulting 3D cubic aperture provides full-space coverage and nearly uniform angular resolution, and thus its synthesized spectrum is adopted as the benchmarking ground truth in the sequel of this paper.
		However, the number of required samples scales as \(N=U^3\), leading to long measurement time and substantial positioning overhead even for synthesizing a single spectrum image.
		
		In the full-space SPS synthesis, the spatial sampling interval is slightly smaller than the conventional half-wavelength spacing \cite{medbo201560}.
		This choice is motivated by two considerations. First, for the 3D cubic sampling grid in (\ref{eq:samp_3D_cubic}), an incoming direction 
		\(\br s=[u,v,w]^T\) induces adjacent-sample phase progressions \(2\pi d u/\lambda\), \(2\pi d v/\lambda\), and \(2\pi d w/\lambda\) along the \(x\)-, \(y\)-, and \(z\)-axes, respectively. 
		If \(d=\lambda/2\), the two antipodal directions \((u,v,w)=(1,0,0)\) and \((-1,0,0)\) produce \(x\)-axis phase increments \(\pi\) and \(-\pi\), which are indistinguishable modulo \(2\pi\), while producing the same zero phase progressions along the \(y\)- and \(z\)-axes. 
		Thus, the corresponding full cubic-array steering responses become identical at these two antipodal boundary points; the same issue also occurs for the antipodal boundary directions along the \(y\)- and \(z\)-axes.
		By choosing \(d<\lambda/2\), all adjacent-sample phase progressions associated with \(u,v,w\in[-1,1]\) lie strictly inside \((-\pi,\pi)\).
		Second, while reducing \(d\) helps avoid this boundary ambiguity, an excessively small spacing would shrink the effective aperture of the sampled region for a fixed \(U\) 
		and degrade angular resolution. Therefore, \(d=0.47\lambda\) is adopted as a compromise that avoids end-fire boundary ambiguity
		while preserving an aperture close to that of the sampling grid with half-wavelength spacing in this paper.

		\subsection{Conventional Half 3D Space Online Synthesis}\label{sec:System_Review}

		\begin{figure}[t]
			\centering
			\includegraphics[width=3.3in]{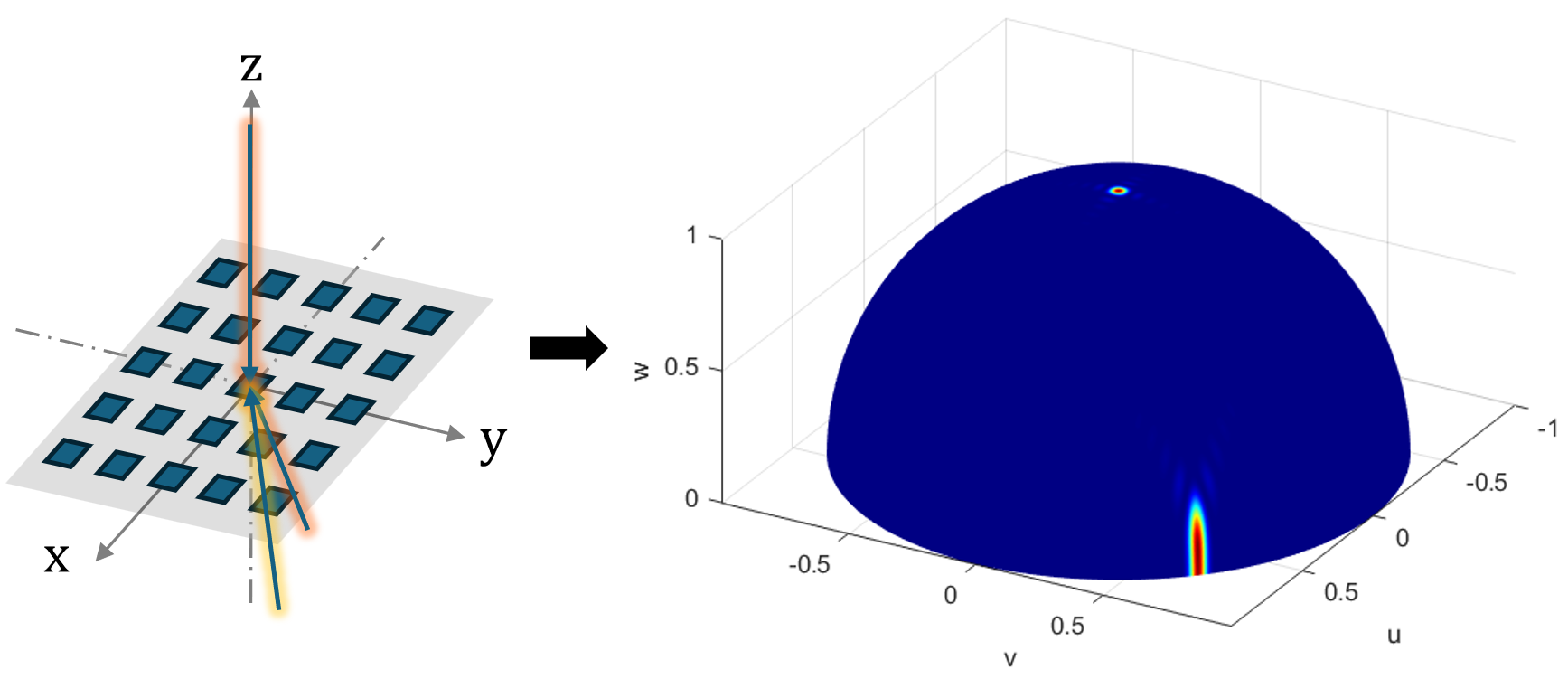}
			\centering 
			\caption{An example of conventional half 3D spatial spectrum synthesized online via UPA and analog receive combining with three MPCs at $(45^{\circ}, 5^{\circ}), (45^{\circ}, 90^{\circ})$, and $(45^{\circ}, -30^{\circ})$, respectively.}
			\label{fig:half_3D_SPP_example}
		\end{figure}

		
		
		
		To synthesize spatial spectrum more efficiently, most existing works apply analog receive combining and set $\{\br{p}_n\}$ to form a UPA \cite{zhao2023nerf2}, i.e.,
		\begin{align}\label{eq:samp_UPA}
			&\br{p}_{n_x,n_y}  = \left[(n_x-l_{U})d,(n_y-l_{U})d,0\right]^T, \\
			\nonumber
			& n_x = 0,1,...,U-1, \enspace n_y = 0,1,...,U-1,
		\end{align}
		where $N = U^2$. In typical analog-only beamforming architectures with a single RF chain, only the beamformed scalar output is acquirable; 
		thus the spectrum is synthesized online by steering the analog phase shifters.
		Discretizing the half-space angular region $(\phi,\theta)$ into $Q_{\text{half}}$ bins,
		to evaluate $\mathbb{P}(\phi,\theta)$ at a grid point $(\phi_j,\theta_j), j = 1,\ldots,Q_{\text{half}}$, the phased array configures its phase shifters to implement 
		weights $\{e^{ -\mr{j} \frac{2 \pi}{\lambda} \br{\nu}^T(\phi_j,\theta_j) \br{p}_n}\}_{n \in \mathcal{N}}$. 
		The power of the RF-chain combining output then directly yields the spectrum value at $(\phi_j, \theta_j)$. 
		Notice that although the single RF chain analog combining implementation requires repeated pilot or sensing signal transmissions during beam sweeping, the associated measurement latency
		 can be small due to fast electronic beam steering. This overhead is often acceptable for massive spatial spectrum acquisition over a large-scale channel 
		 stationarity interval, 
		especially considering the substantial reduction in the hardware cost of the RF chain in return.
		This renders online spectrum synthesis via analog combining a low-complexity and implementation-friendly solution for channel characterization.

		However, a UPA inherently cannot distinguish MPCs arriving from opposite sides of the array plane
		and thus practical implementation provides only half-space (hemisphere) coverage \cite{zhao2023nerf2}. 
		An illustrative example based on (\ref{eq:signal_model_K_discrete}) is shown in Fig. \ref{fig:half_3D_SPP_example}, where an MPC located behind the UPA (e.g., $\theta = -30^{\circ}$) 
		cannot be captured by the hemisphere spatial spectrum.	Moreover, the angular resolution becomes highly non-uniform when $\theta$ is close to zero due to shrinkage of the projected aperture.
		This limitation manifests as the elongated spectrum peak of the MPC at $(45^{\circ}, 5^{\circ})$ in Fig. \ref{fig:half_3D_SPP_example}.

		\section{Rotating ULA-enabled Online Synthesis}\label{sec:Proposed_RULA_CT} 
		
		To enable full-space 3D spectrum online synthesis with high angular resolution while avoiding dense 3D cubic sampling, 
		we jointly design the spatial sampling positions \(\{\br{p}_n\}\) and the spectrum synthesis rule. 
		In particular, instead of applying \eqref{eq:conventional_3D_SPSup} directly, we rotate a ULA to collect partial coherent sums over multiple orientations.
		When only the orientation-wise partial spectrum images are available, we fuse them through a pointwise minimum operator to synthesize the overall SPS.
		When the overall coherent sum across rotations is accessible and phase-consistent across orientations, we further develop a joint processing rule by accumulating all collected coherent sums before spectrum synthesis.
		Both synthesis rules are inspired by the conventional CT reconstruction techniques in the sense that they fuse multiple partial projections to reconstruct the full-space spectrum image.
		For the ease of the discussion, we ignore the AWGN term in the following.
		
		\begin{figure}[t]
			\centering
			\includegraphics[width=3.4in]{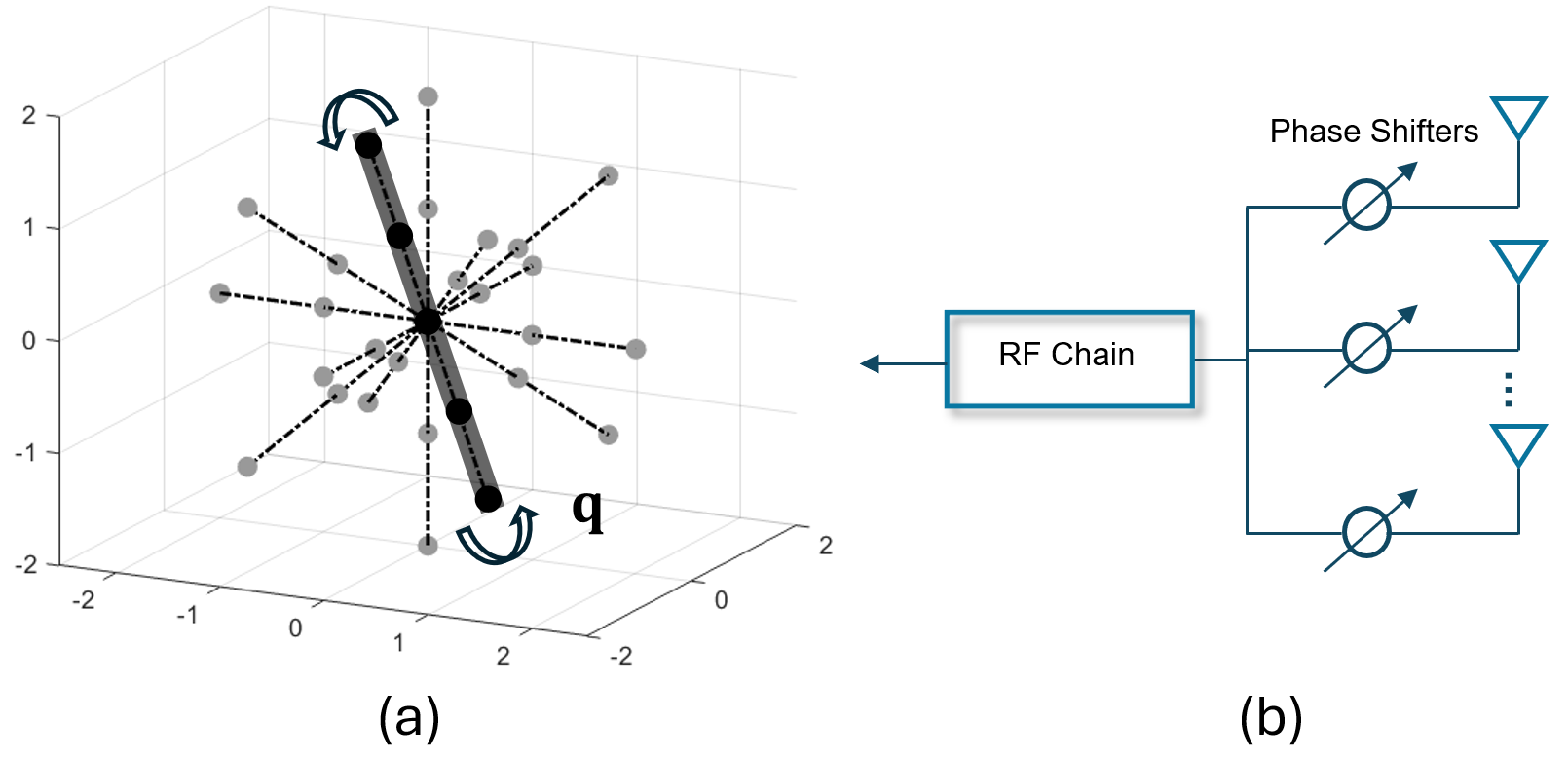}
			\centering
			\caption{(a). Sampling positions $\{\br{p}_n\}$ are obtained through rotating a ULA. (b). Underlying analog-only beamforming architecture.}
			\label{fig:RULA_sampling}
		\end{figure}
		Specifically, as illustrated in Fig. \ref{fig:RULA_sampling}(a), a ULA
		with \(U\) omnidirectional elements is rotated about its center located at the origin of a 3D Cartesian coordinate system. Each rotation is characterized by a unit direction vector as
		\begin{align}\label{eq:q_def}
			\br{q} = \left[q_x,q_y,q_z\right]^T, \enspace \|\br{q}\| = 1,
		\end{align}
		which specifies the orientation of the ULA axis. For example, \(\br{q}=[1,0,0]^T\) corresponds to a ULA aligned with the x-axis.
		At a given $\br{q}$, the sampling positions of the ULA are
		\begin{align}\label{eq:RULA_samp}
			\br{p}_i = \br{q} \left(i-l_U\right) d, \enspace i = 0,1,...,U-1.
		\end{align}
		The received signals along all the \(U\) elements at orientation \(\br q\) are given by
		\begin{align}
			\br r_{\br q}
			=
			\left[
			r(\br q(0-l_U)d),\ldots,
			r(\br q(U-1-l_U)d)
			\right]^T .
		\end{align}
		Substituting the discrete channel model in \eqref{eq:signal_model_K_discrete} into the samples collected along this ULA gives a standard 1D ULA observation. Specifically, define the projected spatial frequency of the \(k\)-th MPC along orientation \(\br q\) as
		\begin{align}\label{eq:mu_k_projection}
			\mu_k(\br q)
			\triangleq
			\br\nu_k^T\br q .
		\end{align}
		Then, the received signal at the \(i\)-th element becomes
		\begin{align}
			r(\br q(i-l_U)d)
			=
			\sum_{k=1}^{K}
			\alpha_k
			e^{\mr j\frac{2\pi}{\lambda}
			\mu_k(\br q)(i-l_U)d},
		\end{align}
		which is exactly the conventional 1D beamforming model with spatial frequencies \(\{\mu_k(\br q)\}\). 
		With \(Q_{\text{full}}\) denoting the number of angular grid points in the \((\phi,\theta)\) domain and \(\br\nu_j=\br\nu(\phi_j,\theta_j)\) the \(j\)-th grid direction, we further define the orientation-dependent 1D steering matrix \(\br A_{\br q}\in\mathbb C^{Q_{\text{full}}\times U}\) as
		\begin{align}\label{eq:RULA_Aq_def}
			[\br A_{\br q}]_{j,i+1}
			=
			e^{-\mr j\frac{2\pi}{\lambda}
			\br\nu_j^T\br q(i-l_U)d},
			\quad i=0,\ldots,U-1 .
		\end{align}
		Based on the analog beamforming architecture shown in Fig.~\ref{fig:RULA_sampling}(b), at each orientation \(\br q\), the ULA steers the beam towards each angular grid point \((\phi_j,\theta_j)\) by setting the analog phase shifters to implement weights \(\{[\br A_{\br q}]_{j,i+1}\}_{i=0}^{U-1}\). It then performs 1D analog receive combining to obtain partial coherent sums \([\br A_{\br q}\br r_{\br q}]_j\) associated with the angular grid point \((\phi_j,\theta_j)\). Depending on whether the overall coherent sum across orientations is accessible or not, we propose two spectrum synthesis rules to yield the overall 3D spatial spectrum, namely the pointwise minimum fusion and the joint coherent fusion, which are detailed in the following.
		
		\subsubsection{Pointwise Minimum Fusion}
		Once the partial coherent sums at each orientation are obtained, the partial spectrum image associated with orientation \(\br q\) is given as
		\begin{align}\label{eq:partial_image}
			\mathcal P_{\br q}(\phi_j,\theta_j)
			=
			\frac{1}{U}
			\left|
			[\br A_{\br q}\br r_{\br q}]_j
			\right|^2,
			\quad j=1,\ldots,Q_{\text{full}}.
		\end{align}
		Equivalently, we have
		\begin{align}
			\mathcal{P}_{\br{q}}(\phi_j,\theta_j) 
			& =  \frac{1}{U} \left|\sum\limits_{i=0}^{U-1}  r(\br{q} \left(i-l_U\right) d) e^{ -\mr{j} \frac{2 \pi}{\lambda} \br{\nu}^T_j \br{q} \left(i-l_U\right) d}\right|^2, \nonumber\\
			j&=1,\ldots,Q_{\text{full}}.
		\end{align}
		Let \(\mathcal P_{\br q}(\phi,\theta)\triangleq\{\mathcal P_{\br q}(\phi_j,\theta_j)\}_{j=1}^{Q_{\text{full}}}\). Compared with the UPA-based online approach in Section~\ref{sec:System_Review}, each \(\mathcal P_{\br q}(\phi,\theta)\) is generated from a ULA 
		with 1D aperture, which yields high per-axis resolution while
		retaining the low-complexity hardware as a single RF chain with $U$ elements.
		\begin{figure*}[t]
				\setlength{\abovecaptionskip}{+0mm}
				\setlength{\belowcaptionskip}{+0mm}
				\centering
				\includegraphics[width=\textwidth]{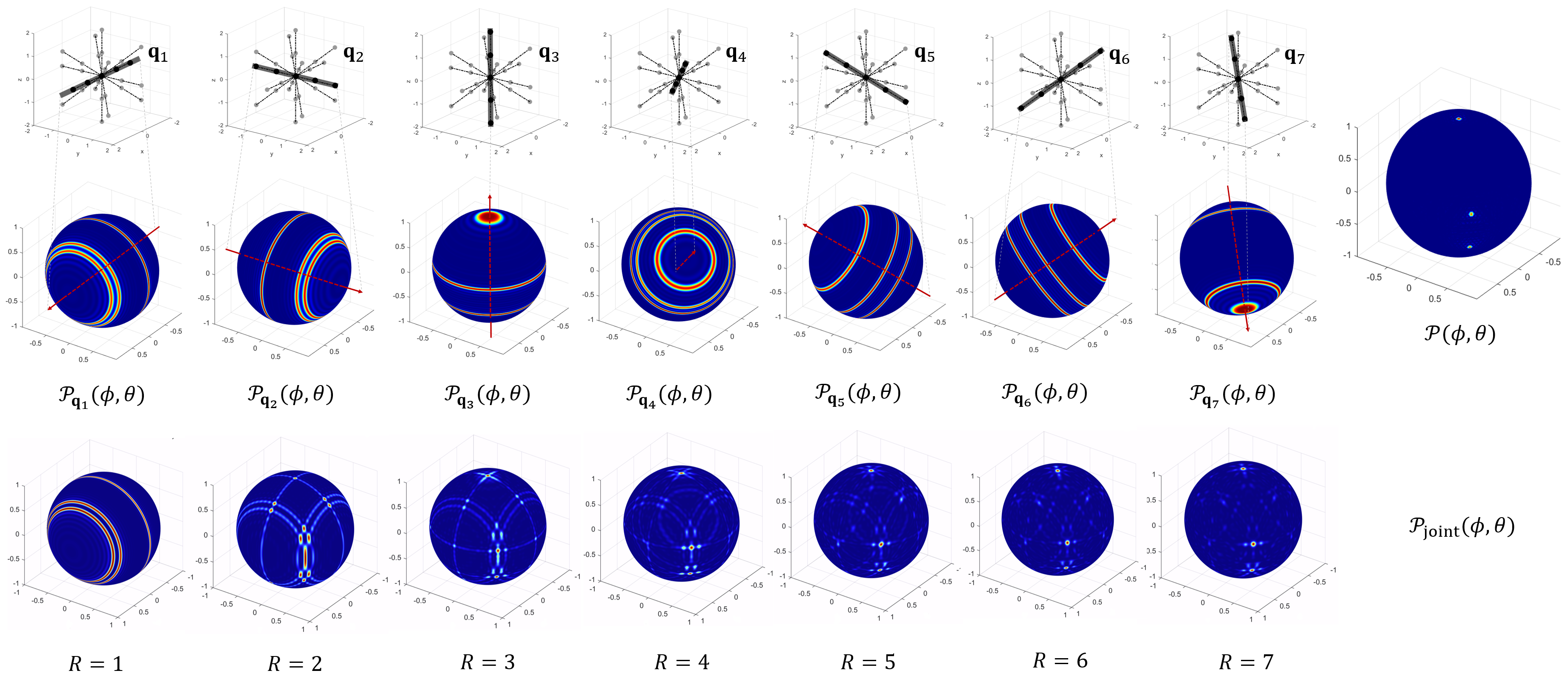}
				\centering
				\caption{Proposed RULA-enabled CT technique for 3D spectrum online synthesis. The top row illustrates the adopted orientations $\mathbf{q}_r$, the middle row shows the corresponding partial slice-like spectra $\mathcal{P}_{\mathbf{q}_r}(\phi,\theta)$ and the pointwise minimum fused spectrum $\mathcal{P}(\phi,\theta)$, and the bottom row demonstrates the progressive joint coherent fusion rule from $R=1$ to $R=7$, yielding the joint coherent fused spectrum $\mathcal{P}_{\mathrm{joint}}(\phi,\theta)$.}
				\label{fig:Joint_CT}
		\end{figure*}
		The overall spectrum is synthesized by fusing the partial images via a pointwise minimum over each $(\phi_j,\theta_j)$ as
		\begin{align}\label{eq:overall_image}
			\mathcal{P}(\phi_j,\theta_j) = \min_{\br{q} \in \mathcal Q} \{\mathcal{P}_{\br{q}}(\phi_j,\theta_j)\}, \enspace j=1,\ldots,Q_{\text{full}},
		\end{align}
		where \(\mathcal Q\triangleq\{\br q_1,\ldots,\br q_R\}\) denotes the set of the adopted $R$ orientations. Intuitively, \(\mathcal P_{\mathbf q}(\phi,\theta)\) contains true MPC peaks that are consistent across orientations, while many sidelobe artifacts are orientation-dependent, hence the minimum operator tends to preserve consistent peaks and suppress inconsistent sidelobes.
		To be specific, under the far-field channel model in \eqref{eq:signal_model_K_discrete}, if the MPCs are well separated and the angular grid is sufficiently fine, 
		then at the true direction \((\phi_k,\theta_k)\) of the \(k\)-th MPC, the phase terms align and yield
		\begin{align}\label{eq:peak_alignment}
			\mathcal P_{\mathbf q}(\phi_k,\theta_k)\approx U|\alpha_k|^2,\quad \forall\,\mathbf q\in\mathcal Q,
		\end{align}
		and therefore \(\mathcal P(\phi_k,\theta_k)\approx U|\alpha_k|^2\), which preserves the relative peak powers of the MPCs. Moreover, this rule only requires the powers of the orientation-wise partial coherent sums, and is therefore insensitive to an unknown orientation-dependent common output phase offset. Specifically, with \(\br A_m\triangleq \br A_{\br q_m}\) and \(\br r_m\triangleq \br r_{\br q_m}\), let \(s_{m,j}\triangleq[\br A_m\br r_m]_j\) denote the ideal partial coherent sum at orientation \(\br q_m\). If the analog-combining chain introduces a common output phase offset \(\psi_m\) at this orientation, the observed coherent sum becomes
		\begin{align}\label{eq:phase_offset}
			\tilde{s}_{m,j}=e^{\mr j\psi_m}s_{m,j}.
		\end{align}
		Since \(\frac{1}{U}|\tilde{s}_{m,j}|^2=\frac{1}{U}|s_{m,j}|^2\), such an offset is removed before fusing different orientations.

		\subsubsection{Joint Coherent Fusion}
		If the common phase shift can be estimated and calibrated, the overall coherent sum across orientations is then accessible.
		Towards this end, the joint coherent fusion can be adopted. Specifically, define
		\begin{align}\label{eq:joint_A_r}
			\br y_{\rm joint}
			=
			\sum_{m=1}^{R}\br A_m\br r_m
			-
			(R-1)r_c\br 1,
		\end{align}
		where \(r_c\) is the common center-position sample shared by all rotated ULAs. The subtraction avoids counting the center sample repeatedly. 
		The corresponding joint-processing SPS is
		\begin{align}\label{eq:joint_SPS}
			\mathcal P_{\rm joint}(\phi_j,\theta_j)
			=
			\frac{1}{RU}
			\left|
			[\br y_{\rm joint}]_j
			\right|^2,
			\quad j=1,\ldots,Q_{\text{full}}.
		\end{align}
		Unlike the pointwise minimum fusion rule, this joint fusion rule requires coherent addition across different orientations. Therefore, unknown inter-orientation phase shifts can distort the accumulated coherent sum. However, this coherent fusion yields better synthesized spectrum when the SNR is low, as it effectively combines the signal power across orientations before taking the magnitude square.

		An illustrative example with 7 orientations based on the minimum fusion rule in (\ref{eq:overall_image}) and the joint coherent fusion rule in (\ref{eq:joint_SPS}) is provided in Fig.~\ref{fig:Joint_CT} 
		with MPCs being the same as previous example to facilitate comparison.
		Notice that each partial 3D image in the second row is fundamentally a 1D spatial spectrum lifted onto the 3D angular sphere through the projection \(\mu=\br\nu^T(\phi,\theta)\br q\).

		To conclude this section, it is worth comparing different practical implementation architectures for the aforementioned sampling geometries, 
		as summarized in Table~\ref{tab:implementation_arch}. 
		For the 3D cubic sampling geometry, the single-antenna implementation is feasible by sequentially moving one antenna element over all \(U^3\) 
		positions, which is exactly a virtual-array measurement approach \cite{medbo2016frequency}. 
		However, implementing it as a physical fully-digital or analog-combining array would require constructing a gigantic \(U\times U\times U\) 
		3D antenna array, together with \(U^3\) RF chains or phase-shifter branches, which is generally impractical.
		For the proposed RULA sampling geometry, it is worth noting that the overall SPS acquisition time is typically dominated by mechanical ULA orientation switching rather than per-orientation beam sweeping or received signal measurement.
		Thus, the proposed analog combining implementation can achieve a similar overall SPS acquisition time to the fully digital implementation when mechanical orientation switching is the dominant bottleneck, while significantly reducing the number of RF chains from \(U\) to 1.

		\begin{table*}[t]
			\centering
			\caption{Implementation feasibility and requirements of different spatial sampling strategies.}
			\label{tab:implementation_arch}
			\renewcommand{\arraystretch}{1.18}
			\scriptsize
			\begin{tabular}{|c|c|c|c|}
				\hline
				\textbf{Sampling Geometry} &
				\textbf{\shortstack{Single RF chain, single antenna}} &
				\textbf{\shortstack{Fully digital, multi-antenna}} &
				\textbf{\shortstack{Analog combining, multi-antenna}} \\
				\hline
				\textbf{3D cubic} &
				\(\checkmark\) \shortstack{Sequential virtual array\\Antenna movement: \(U^3\)} &
				\(\times\) \shortstack{Requires physical \(U^3\)-antenna array\\and \(U^3\) RF chains} &
				\(\times\) \shortstack{Requires physical \(U^3\)-antenna array\\and \(U^3\) analog branches} \\
				\hline
				\textbf{UPA} &
				\(\checkmark\) \shortstack{Sequential virtual array\\Antenna movement: \(U^2\)} &
				\(\checkmark\) \shortstack{\(U^2\) antennas, \(U^2\) RF chains\\Antenna movement: 0} &
				\(\checkmark\) \shortstack{\(U^2\) antennas, single RF chain\\Antenna movement: 0} \\
				\hline
				\textbf{RULA} &
				\(\checkmark\) \shortstack{Sequential virtual array\\Antenna movement: \(RU\)} &
				\(\checkmark\) \shortstack{\(U\) antennas, \(U\) RF chains\\Antenna movement: \(R\)} &
				\(\checkmark\) \shortstack{\(U\) antennas, single RF chain\\Antenna movement: \(R\)} \\
				\hline
			\end{tabular}
		\end{table*}
		
		\begin{remark}
			\emph{The design intuition is that, for a fixed number of sampled positions, a ULA maximizes angular resolution
			 along its 1D aperture. Full 3D coverage is then achieved by rotating the ULA to probe multiple aperture orientations.
			  As the number of orientations increases, the sampling set approaches a densely sampled ball-like 3D volume, which is 
			  similar to a virtual 3D cubic array. The practically relevant question is therefore how to design a \emph{small} 
			  set \(\mathcal Q\) that can faithfully recover the full-space spectrum while keeping the sampling overhead low, which is addressed in detail in the next section.}
		\end{remark}
		

		\section{Environment-agnostic Orientation Design}\label{sec:metric_initial_orientation_design}

		The orientation design should be environment-agnostic in the sense that it should not rely on any prior information about the surrounding wireless environment or any specific channel realization.
		As shown in Section~\ref{sec:Proposed_RULA_CT}, each orientation distinguishes MPCs through their projected spatial frequencies in \eqref{eq:mu_k_projection}.
		Therefore, to design a small set \(\mathcal Q\) in a principled manner, a natural criterion is to make potentially close MPC pairs well separated across the selected projections.
		Towards this end, we propose the following primary goal and secondary goal for the orientation design.
		The primary goal is to maximize the separability of possible MPC pairs in the projected 1D spatial-frequency domain.
		After this primary separation goal is met, the secondary goal is to reduce redundancy of the orientation set when measuring each individual MPC.

		\subsection{Maximizing Expected Worst-case Squared Projected Separation}

		Specifically, in the primary goal, we model the MPC direction vectors in \eqref{eq:nu_ori_def} as random unit vectors and denote them by \(\br s_k\equiv\br\nu_k=[u_k,v_k,w_k]^T\), \(k=1,\ldots,K\), for this statistical design analysis.
		Suppose that \(\{\br s_k\}_{k=1}^{K}\) are independently and uniformly distributed on the unit sphere.
		For a candidate orientation set \(\mathcal Q=\{\br q_i\}_{i=1}^{R}\), define the squared projected separation between two MPCs as
		\begin{align}\label{eq:orientation_pair_distance}
			d_{k\ell}^2(\mathcal Q)
			\triangleq
			\sum_{i=1}^{R}
			\left(\br q_i^T(\br s_k-\br s_\ell)\right)^2,\quad k\ne \ell .
		\end{align}
		This is exactly the distance induced by the collection of 1D projections measured by the rotated ULA. 
		If \(d_{k\ell}^2(\mathcal Q)\) is small, then the two directions have nearly identical projected spatial frequencies over all adopted orientations and are therefore difficult to distinguish.
		In particular, when \(\br q_i^T\br s_k=\br q_i^T\br s_\ell\), these two MPCs are indistinguishable along the \(i\)-th orientation, and as a result, \(\left(\br q_i^T(\br s_k-\br s_\ell)\right)^2=0\), as shown in Fig.~\ref{fig:equal_projection_ring}.
		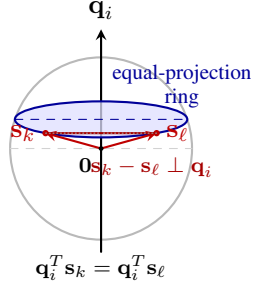
\begin{figure}[t]
			\centering
			\begin{tikzpicture}[line cap=round,line join=round,>=stealth,scale=0.70]
				\coordinate (O) at (0,0);
				\coordinate (Qtop) at (0,2.25);
				\coordinate (Qbot) at (0,-1.95);
				\coordinate (C) at (0,0.55);
				\coordinate (Sk) at (-1.05,0.29);
				\coordinate (Sl) at (1.05,0.29);
				
				\draw[gray!55,thick] (O) circle (1.72);
				\draw[gray!35,dashed] (-1.72,0) -- (1.72,0);
				\draw[->,thick] (Qbot) -- (Qtop) node[above] {$\br q_i$};
				
				\fill[blue!12,opacity=0.75] (C) ellipse (1.63 and 0.34);
				\draw[blue!60!black,thick] (C) ellipse (1.63 and 0.34);
				\draw[blue!60!black,dashed] (-1.63,0.55) -- (1.63,0.55);
				\node[blue!60!black,align=center] at (1.53,1.18) {\footnotesize equal-projection\\[-1mm]\footnotesize ring};
				
				\draw[->,red!75!black,thick] (O) -- (Sk) node[left] {$\br s_k$};
				\draw[->,red!75!black,thick] (O) -- (Sl) node[right] {$\br s_\ell$};
				\fill[red!75!black] (Sk) circle (1.7pt);
				\fill[red!75!black] (Sl) circle (1.7pt);
				
				\draw[red!65!black,thick] (Sk) -- (Sl);
				\draw[densely dotted,gray!70] (Sk) -- (0,0.29);
				\draw[densely dotted,gray!70] (Sl) -- (0,0.29);
				\fill[black] (O) circle (1.3pt) node[below left] {\footnotesize $\br 0$};
				\node[red!65!black,align=left] at (0.96,-0.35) {\footnotesize $\br s_k-\br s_\ell\perp \br q_i$};
				\node[align=center] at (0,-2.25) {\footnotesize $\br q_i^T\br s_k=\br q_i^T\br s_\ell$};
			\end{tikzpicture}
			\caption{Illustration of two MPC directions having identical projection along orientation \(\br q_i\). Points with the same value of \(\br q_i^T\br s\) lie on a ring obtained by intersecting the unit sphere with a plane perpendicular to \(\br q_i\). Hence \(\br q_i^T(\br s_k-\br s_\ell)=0\), giving zero squared projected separation along this orientation.}
			\label{fig:equal_projection_ring}
			\vspace{-10pt}
		\end{figure}
		We then design \(\mathcal Q\) by the max-min criterion as
		\begin{align}\label{eq:init_orientation_design}
			\mathcal Q^{\star}
			=
			\arg\max_{\{\br q_i:\|\br q_i\|=1\}}
			\mathbb E_{\{\br s_k\}}
			\left[
			\min_{k\ne \ell} d_{k\ell}^2(\mathcal Q)
			\right],
		\end{align}
		where the expectation is taken over $\{\br s_k\} \triangleq \br s_1  \times ... \times \br s_K$, 
		where each $\br s_k$ is independently and uniformly distributed on the unit sphere \(\mathcal S^2\), i.e., \(\br s_k\sim \operatorname{Unif}(\mathcal S^2)\).
		 Equivalently, if \(\br g_k\sim\mathcal{N}(\br 0,\br I_3)\), then \(\br s_k=\br g_k/\|\br g_k\|\). This isotropic distribution satisfies
		\begin{align}
			\mathbb E[\br s_k]=\br 0,\qquad
			\mathbb E[\br s_k\br s_k^T]=\frac{1}{3}\br I_3 .
		\end{align}
		Since \(\br q\) and \(-\br q\) correspond to the same physical ULA axis, the candidate pool is generated over one
		hemisphere, with a fixed sign convention to avoid duplicate axes.

		To solve the problem in (\ref{eq:init_orientation_design}), let
		\begin{align}
			\br G(\mathcal Q)=\sum_{i=1}^{R}\br q_i\br q_i^T .
		\end{align}
		Then
		\begin{align}
			d_{k\ell}^2(\mathcal Q)
			=
			(\br s_k-\br s_\ell)^T
			\br G(\mathcal Q)
			(\br s_k-\br s_\ell).
		\end{align}
		For any feasible orientation set, \(\br G(\mathcal Q)\succeq \br 0\) and
		\begin{align}
			\operatorname{tr}\{\br G(\mathcal Q)\}
			=
			\sum_{i=1}^{R}\|\br q_i\|^2
			=
			R.
		\end{align}
		This motivates a semidefinite relaxation of \eqref{eq:init_orientation_design}. 
		Specifically, instead of optimizing directly over the finite unit-vector decomposition \(\br G=\sum_{i=1}^{R}\br q_i\br q_i^T\), 
		we keep only the necessary conditions \(\br G\succeq \br 0\) and \(\operatorname{tr}\{\br G\}=R\). 
		We then have the following relaxed optimization problem as
		\begin{align}\label{eq:relaxed_G_design}
			\max_{\br G\succeq \br 0}\quad
			F(\br G)
			\quad
			\text{s.t.}\quad
			\operatorname{tr}\{\br G\}=R,
		\end{align}
		where
		\begin{align}
			F(\br G)
			=
			\mathbb E_{\{\br s_k\}}
			\left[
			\min_{k\ne \ell}
			(\br s_k-\br s_\ell)^T
			\br G
			(\br s_k-\br s_\ell)
			\right].
		\end{align}
		The following proposition identifies the global optimal solution of \eqref{eq:relaxed_G_design}.
		\begin{proposition}\label{prop:isotropic_G_optimal}
			\emph{For the relaxed problem in \eqref{eq:relaxed_G_design}, a globally optimal solution is given as:}
			\begin{align}\label{eq:isotropic_G_optimal}
				\br G^\star
				=
				\frac{R}{3}\br I_3 .
			\end{align}
		\end{proposition}
		\begin{proof}
			Please refer to Appendix~\ref{app:proof_isotropic_G}.
		\end{proof}

		The remaining problem is how to decompose \(\br G^\star\) into a finite unit-vector sum \(\sum_{i=1}^{R}\br q_i\br q_i^T\). It turns out that, as long as \(R\ge 3\), the isotropic matrix in \eqref{eq:isotropic_G_optimal} can always be exactly realized by a finite orientation set.
		\begin{proposition}\label{prop:tight_frame_existence}
			\emph{For any integer \(R\ge 3\), there exists a set of unit-norm orientations \(\mathcal Q=\{\br q_i\}_{i=1}^{R}\) such that}
			\begin{align}\label{eq:tight_frame_existence}
				\sum_{i=1}^{R}\br q_i\br q_i^T
				=
				\frac{R}{3}\br I_3 .
			\end{align}
		\end{proposition}
		\begin{proof}
			Please refer to Appendix~\ref{app:proof_tight_frame_existence}.
		\end{proof}
		Therefore, the solution in \eqref{eq:isotropic_G_optimal} is not merely the global optimum of a relaxed problem. When \(R\ge 3\), it is also attainable by an actual finite unit-norm orientation set. Such a set is also known as a unit-norm tight frame in \(\mathbb R^3\).

		The corresponding optimal objective value also admits a useful interpretation, summarized as follows.
		\begin{proposition}\label{prop:isotropic_objective_value}
			\emph{Let \(
				D_{\min}
				\triangleq
				\min_{k\ne \ell}
				\|\br s_k-\br s_\ell\|^2\)}
			\emph{denote the minimum squared Euclidean distance among the \(K\) random MPC directions. Then}
			\emph{for \(K=2\),}
			\begin{align}
				F(\br G^\star)=\frac{2R}{3}.
			\end{align}
			\emph{Moreover, with \(N_p=K(K-1)/2\), for any fixed \(x\ge 0\),}
			\begin{align}
				\Pr(N_pD_{\min}>x)
				\rightarrow
				e^{-x/4},
				\quad K\rightarrow\infty,
			\end{align}
			\emph{or equivalently \(N_pD_{\min}\) converges in distribution to an exponential random variable with rate \(1/4\). Consequently, for sufficiently large \(K\),}
			\begin{align}\label{eq:isotropic_F_largeK}
				F(\br G^\star)
				\approx
				\frac{8R}{3K(K-1)}.
			\end{align}
		\end{proposition}
		\begin{proof}
			Please refer to Appendix~\ref{app:proof_isotropic_objective_value}.
		\end{proof}
		Equation \eqref{eq:isotropic_F_largeK} provides an intuitive scaling law for the orientation design. The optimal objective value $F(\br G^\star)$ increases linearly with \(R\), since each additional orientation contributes one more projection dimension and enlarges the aggregate projected separation between MPCs. 
		In contrast, the objective decreases approximately as \(1/K^2\) when the number of MPCs grows. For $K$ MPCs, there exists \(K(K-1)/2\) MPC pairs, making it increasingly likely that at least
		 one pair is very close on the sphere. 
		Equation \eqref{eq:isotropic_G_optimal} shows that a well-balanced orientation set should make \(\br G(\mathcal Q)\) 
		as close as possible to \(\frac{R}{3}\br I_3\). For \(R<3\), this target is impossible since \(\operatorname{rank}\{\br G(\mathcal Q)\}\le R\). This case is trivial and is thus not our focus.
		For \(R\ge 3\), Proposition~\ref{prop:tight_frame_existence} guarantees the existence of a finite 
		 orientation set that achieves the isotropic target, and thus the relaxed optimum is attainable. 
		 Nevertheless, consider the case where we repeat the three coordinate axes twice to yield \(\br G(\mathcal Q)=2\br I_3\) 
		 when \(R=6\), it is clear that this trivial design does not provide six distinct orientations. In the following,
		 we propose a second design goal to avoid such redundant orientation design that is indistinguishable at the relaxed matrix level.


		\subsection{Minimizing Worst-case Projective Correlation}
		 The second design goal plays a different role from the primary objective. Once \(\br G(\mathcal Q)=\frac{R}{3}\br I_3\) is fixed, all orientation realizations satisfying this isotropic-matrix condition yield the same aggregate projected separation value for every MPC pair. 
		 The remaining degree of freedom lies in how the individual 1D projections are arranged. For an isotropically distributed direction \(\br s\), the projected frequencies observed by two orientations are \(\br q_i^T\br s\) and \(\br q_j^T\br s\), whose correlation is
		\begin{align}
			\frac{
			\mathbb E[(\br q_i^T\br s)(\br q_j^T\br s)]
			}{
			\sqrt{
			\mathbb E[(\br q_i^T\br s)^2]
			\mathbb E[(\br q_j^T\br s)^2]
			}}
			=
			\br q_i^T\br q_j,
		\end{align}
		where we use \(\mathbb E[\br s\br s^T]=\frac{1}{3}\br I_3\). Thus, \(|\br q_i^T\br q_j|^2\) is the squared correlation between the two projected frequencies observed by orientations \(\br q_i\) and \(\br q_j\). If \(|\br q_i^T\br q_j|^2\) is close to 1, then the two orientations are nearly parallel and thus give highly redundant measurements that should be avoided. We therefore minimize the worst-case projective correlation to reduce the largest measurement redundancy subject to the already optimized isotropic-matrix constraint and the orientation unit-norm constraints, which is formulated as
		\begin{subequations}\label{eq:two_layer_orientation_design}
			\begin{align}
				(\text{P1}): \min_{\{\br q_i\}}
				\quad&
				\mu(\mathcal Q)
				\triangleq
				\max_{i\ne j}|\br q_i^T\br q_j|^2
				\label{eq:two_layer_orientation_design_obj}\\
				\text{s.t.}\quad&
				\sum_{i=1}^{R}\br q_i\br q_i^T
				=
				\frac{R}{3}\br I_3,
				\label{eq:two_layer_orientation_design_tf}\\
				&
				\|\br q_i\|=1,\quad i=1,\ldots,R.
				\label{eq:two_layer_orientation_design_unit}
			\end{align}
		\end{subequations}
		The absolute inner product is used because \(\br q\) and \(-\br q\) correspond to the same physical ULA axis. Proposition~\ref{prop:tight_frame_existence} guarantees that the isotropic-matrix constraint in \eqref{eq:two_layer_orientation_design_tf} is feasible for every \(R\ge 3\). This second design criterion therefore preserves the global optimality interpretation of the isotropic relaxed matrix while preventing redundant measurements.


		\begin{algorithm}[t]
			\caption{Multi-start optimization for principled orientation design}
			\label{alg:initial_orientation_optimization}
			\begin{algorithmic}[1]
				\STATE Set \(R\), smoothing parameter \(\beta\), penalty weight \(\rho\), number of random starts \(N_{\rm start}\), and isotropic-feasibility tolerance \(\epsilon_{\rm iso}\).
				\STATE Generate initial angle vectors \(\{(\vartheta_i,\varphi_i)\}_{i=1}^{R}\) from random unit vectors.
				\FOR{each initialization}
					\STATE Minimize \(\mathcal L(\mathcal Q)\) in \eqref{eq:smooth_second_layer_obj} over the spherical angles by a local nonlinear optimizer until convergence or a prescribed iteration limit is reached.
					\STATE Recover the converged candidate \(\mathcal Q\), compute its isotropic violation \(e_{\rm iso}(\mathcal Q)=\|\br G(\mathcal Q)-\frac{R}{3}\br I_3\|_F\), and compute its true projective correlation \(\mu(\mathcal Q)=\max_{i\ne j}|\br q_i^T\br q_j|^2\).
				\ENDFOR
				\STATE Form the admissible candidate set \(\mathcal C=\{\mathcal Q:e_{\rm iso}(\mathcal Q)\le\epsilon_{\rm iso}\}\).
				\STATE Select \(\mathcal Q^\star=\arg\min_{\mathcal Q\in\mathcal C}\mu(\mathcal Q)\).
			\end{algorithmic}
		\end{algorithm}

		Although (P1) is low-dimensional, it is non-convex because of the unit-sphere constraints, the quadratic isotropic-matrix constraint, and the quartic projective-correlation objective. In practice, it can be solved effectively by a multi-start smooth minimax optimization. Parameterize each orientation by two spherical angles,
		\begin{align}
			\br q_i(\vartheta_i,\varphi_i)
			=
			[\sin\vartheta_i\cos\varphi_i,
			\sin\vartheta_i\sin\varphi_i,
			\cos\vartheta_i]^T .
		\end{align}
		The optimization variables are therefore the \(2R\) angle variables \(\{(\vartheta_i,\varphi_i)\}_{i=1}^{R}\), and the unit-norm constraints in \eqref{eq:two_layer_orientation_design_unit} are automatically satisfied by this parameterization.
		For a given candidate set \(\mathcal Q\), replace the maximum in \eqref{eq:two_layer_orientation_design} by the log-sum-exp smooth approximation. Specifically, for \(a_{ij}=|\br q_i^T\br q_j|^2\) and \(M=R(R-1)/2\), we have
		\begin{align}
			\max_{i<j}a_{ij}
			\le
			\frac{1}{\beta}
			\log
			\sum_{i<j}\exp(\beta a_{ij})
			\le
			\max_{i<j}a_{ij}
			+
			\frac{\log M}{\beta},
		\end{align}
		which becomes tight as \(\beta\) increases. The isotropic-matrix condition is then enforced through a quadratic penalty, leading to the following smooth objective function:
		\begin{align}\label{eq:smooth_second_layer_obj}
			\mathcal L(\mathcal Q)
			=&
			\frac{1}{\beta}
			\log
			\sum_{i<j}
			\exp\left(
			\beta |\br q_i^T\br q_j|^2
			\right)
			\nonumber\\
			&+
			\rho
			\left\|
			\sum_{i=1}^{R}\br q_i\br q_i^T
			-
			\frac{R}{3}\br I_3
			\right\|_F^2,
		\end{align}
		where \(\beta>0\) controls the smoothness of the max approximation and \(\rho>0\) controls the isotropic-matrix penalty.
		For large \(\beta\), the log-sum-exp term closely approximates the maximum correlation objective, while for large \(\rho\), the penalty term strongly enforces the isotropic-matrix constraint.
		Hence, minimizing \eqref{eq:smooth_second_layer_obj} over the spherical angle variables amounts to solving a smooth penalized approximation of the original problem (P1). Each local run is performed by the Nelder--Mead simplex method \cite{lagarias1998convergence} over the \(2R\) angle variables.
		The resulting procedure is summarized in Algorithm~\ref{alg:initial_orientation_optimization}.

		\section{Numerical Results}\label{sec:numerical_result}
		
		In this section, we provide numerical results to validate the proposed design.
		Without loss of generality, we set $K = 12$ and all the MPCs are distributed randomly in the full 3D space. The complex coefficients are randomly generated as $\alpha_k = |\alpha_k| e^{\mr{j} \angle \alpha_k}$ with $|\alpha_k| \sim \mathcal{U}(0.8,1)$ and $\angle \alpha_k \sim \mathcal{U}(0,2\pi)$.
		The SNR is defined as \(\text{SNR} = \mathbb{E}\left[|r(\br{p}_n)|^2\right]/\sigma^2.\)
		Besides the proposed RULA-CT with $R = 7$ rotations, we consider the following benchmarks for comparison:
		\begin{itemize}
			\item \textbf{3D Cubic Beamforming \cite{medbo201560}}: With $U^3$ sampling positions specified in (\ref{eq:samp_3D_cubic}), the spectrum $\mathbb{P}(\phi,\theta)$ is synthesized offline via beamforming in (\ref{eq:conventional_3D_SPSup}). This approach approximates $P(\br{\nu})$ well via a Riemann-sum discretization over a volumetric aperture and is thus adopted as the benchmarking ground truth.
			\item \textbf{UPA Beamforming \cite{zhao2023nerf2}}: The spatial sampling positions are specified in (\ref{eq:samp_UPA}) with $U \times U$ UPA, and the spectrum $\mathbb{P}(\phi,\theta)$ is synthesized online via analog receive combining according to (\ref{eq:conventional_3D_SPSup}) with $\theta \in \left[0,\pi/2\right]$ as it can only cover half 3D space without ambiguity.
		\end{itemize}

		\begin{figure}[t]
			\centering
			\includegraphics[width=2.7in]{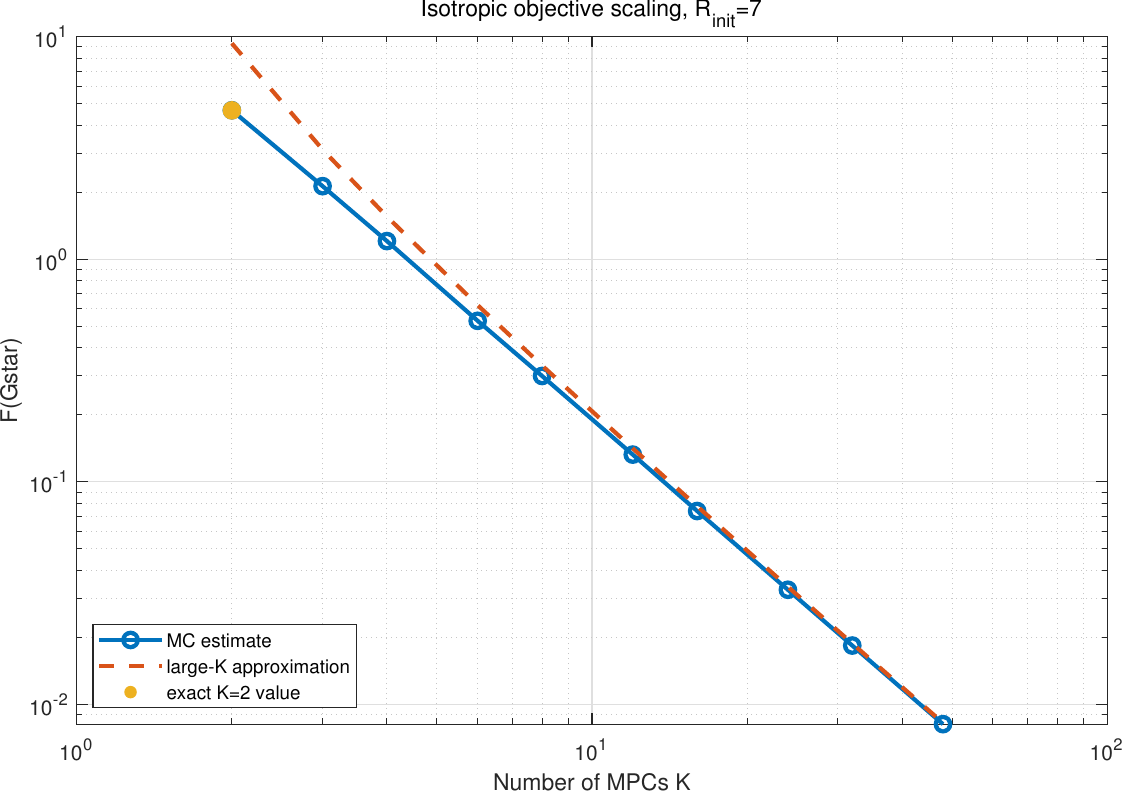}
			\caption{MC verification of the asymptotic scaling of \(F(\br G^\star)\).}
			\label{fig:isotropic_objective_asymptotic}
		\end{figure}

		\subsection{Monte Carlo Verification of the Large-\(K\) Scaling}
		We first verify the asymptotic expression in \eqref{eq:isotropic_F_largeK}. For the isotropic matrix \(\br G^\star=\frac{R}{3}\br I_3\), we estimate
		\begin{align}
			F(\br G^\star)
			=
			\mathbb E
			\left[
			\min_{k\ne \ell}
			(\br s_k-\br s_\ell)^T
			\br G^\star
			(\br s_k-\br s_\ell)
			\right]
		\end{align}
		by \(5\times 10^4\) Monte Carlo (MC) trials for each value of \(K\) with \(R=7\). Fig.~\ref{fig:isotropic_objective_asymptotic} compares the MC estimates with the large-\(K\) approximation in \eqref{eq:isotropic_F_largeK}. The approximation becomes increasingly accurate as \(K\) grows, which is consistent with Proposition \ref{prop:isotropic_objective_value}. The exact value for \(K=2\), namely \(F(\br G^\star)=2R/3\), is also shown for reference.

		\begin{figure*}[t]
			\centering
			\includegraphics[width=\textwidth*12/14]{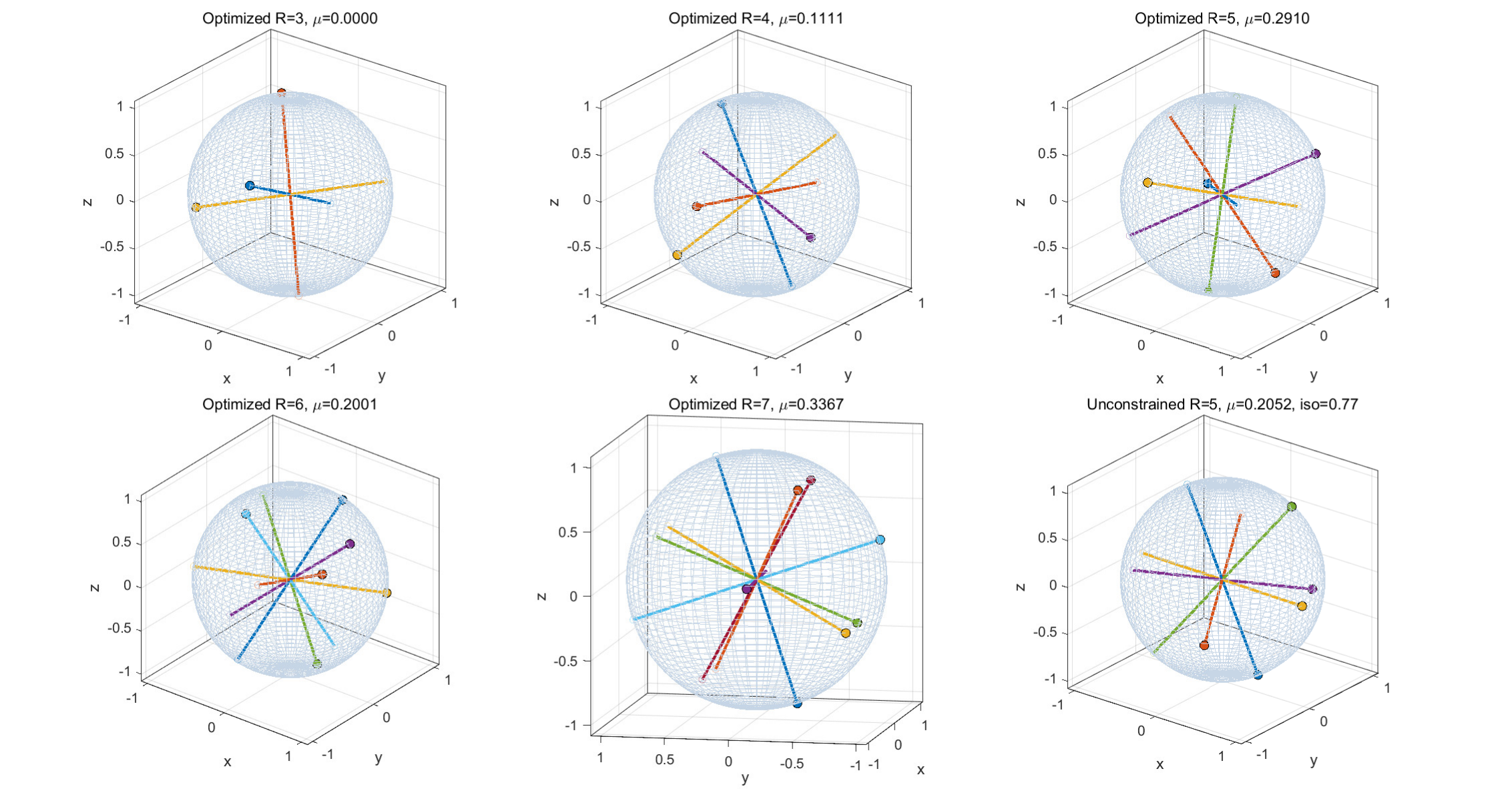}
			\caption{Optimized orientation realizations for \(R=3,\ldots,7\) obtained by solving (P1). 
			Each line represents a ULA orientation.}
			\label{fig:Rinit_orientation_designs}
		\end{figure*}

		\subsection{Numerical Verification of Orientation Realization}
		We next numerically solve (P1) for different $R$ with \(R \ge 3\). The optimization uses the smooth objective $\mathcal L(\mathcal Q)$ in \eqref{eq:smooth_second_layer_obj} with multiple random initializations $N_{\text{start}}$. Fig.~\ref{fig:Rinit_orientation_designs} shows the optimized projective axes for \(R=3,\ldots,7\). The obtained axes are well spread over the sphere while satisfying the isotropic-matrix constraint, which verifies that the isotropic solution in \eqref{eq:isotropic_G_optimal} can be realized by concrete finite orientation sets. Notice that the design
		with $R=7$  is very close to the original 7-orientation design in our preliminary work\cite{hua2026rotating}. For comparison, Fig.~\ref{fig:Rinit_orientation_designs} also includes an unconstrained \(R=5\) solution that minimizes the worst-case projective correlation without enforcing \eqref{eq:two_layer_orientation_design_tf}. This unconstrained design achieves a smaller correlation value $\mu(\mathcal Q)$, but its isotropy error $e_{\text{iso}}(\mathcal Q)$ defined in Algorithm \ref{alg:initial_orientation_optimization} is larger, which shows that the proposed two criteria are indeed different and complementary.
		
		Fig.~\ref{fig:Rinit_orientation_convergence} depicts the convergence behavior of the multi-start smooth minimax optimization. Here, one iteration refers to one local Nelder--Mead simplex update, during which the simplex is modified through reflection, expansion, contraction, or shrinkage according to the evaluated values of \(\mathcal L(\mathcal Q)\) \cite{lagarias1998convergence}. 
		Gray curves correspond to different random starts, while the colored curve gives the selected feasible run with the smallest worst-case projective correlation among the admissible candidate sets. The curves show that the proposed algorithm reliably decreases the penalized objective and recovers low-redundancy realizations under the constraints. We then use the obtained design with \(R=7\) for spectrum synthesis in the following.

		\begin{figure*}[t]
			\centering
			\includegraphics[width=\textwidth*3/4]{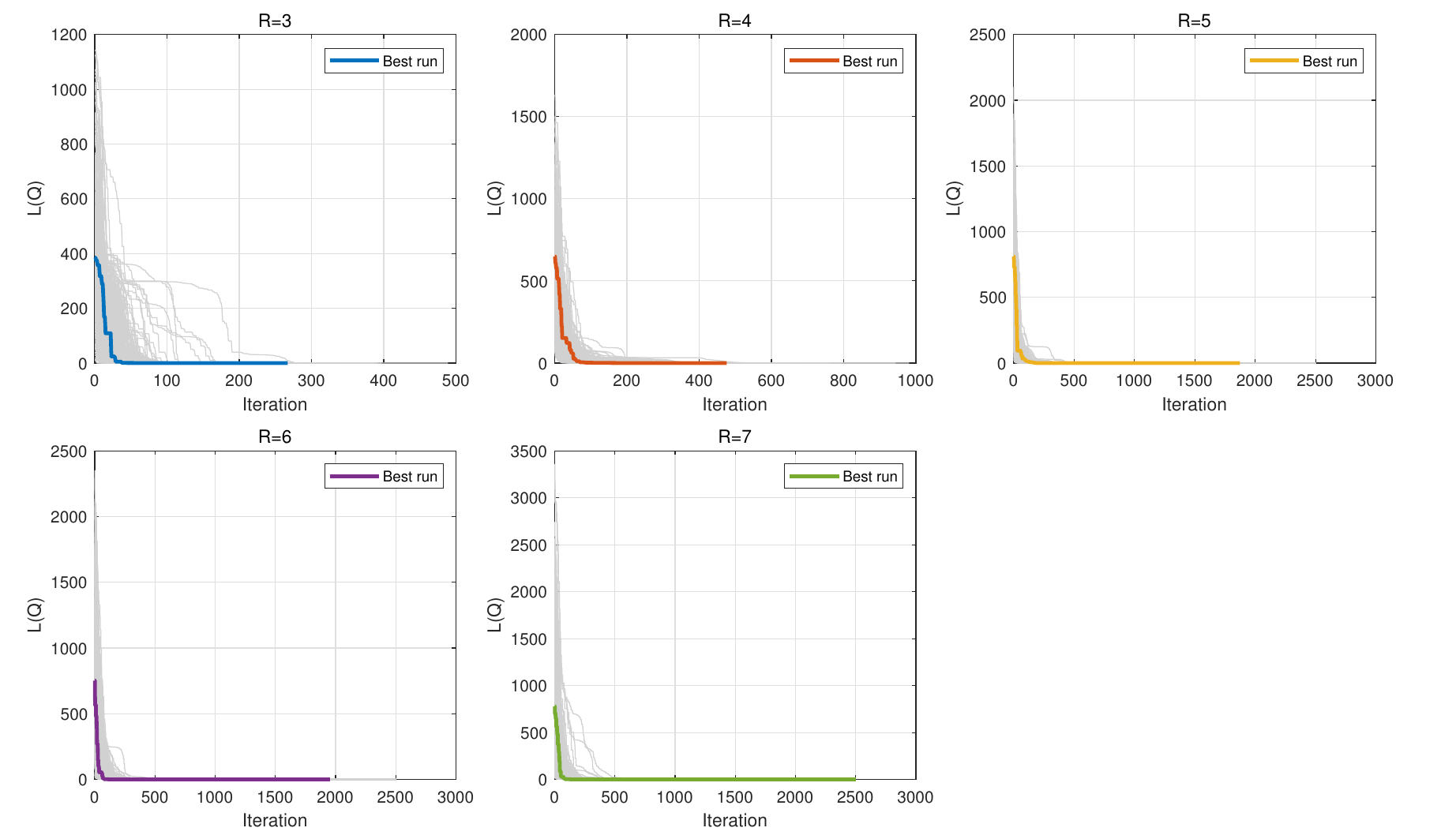}
			\caption{Convergence curves of the proposed multi-start smooth minimax optimization.}
			\label{fig:Rinit_orientation_convergence}
		\end{figure*}
		
		\begin{figure*}[t]
				\centering
				\setlength{\abovecaptionskip}{+0mm}
				\setlength{\belowcaptionskip}{+0mm}
				\subfigure[3D cubic beamforming.]{ \label{fig:3D_cubic_beam}
					\includegraphics[width=1.706in]{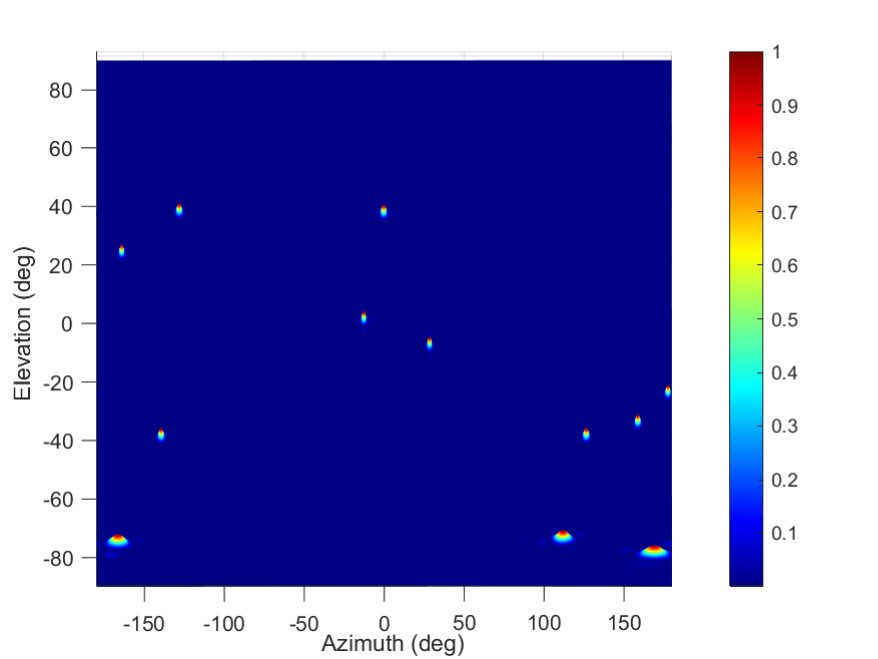}}
				\subfigure[UPA beamforming.]{ \label{fig:UPA_beam}
					\includegraphics[width=1.706in]{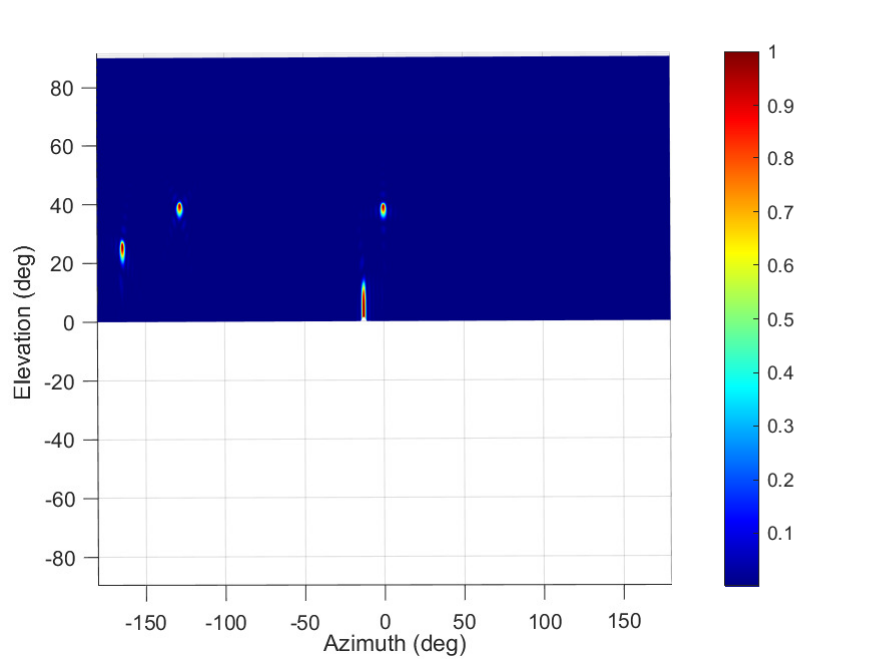}}
				\subfigure[RULA-min with $R=7$.]{ \label{fig:RULA_CT_beam}
					\includegraphics[width=1.706in]{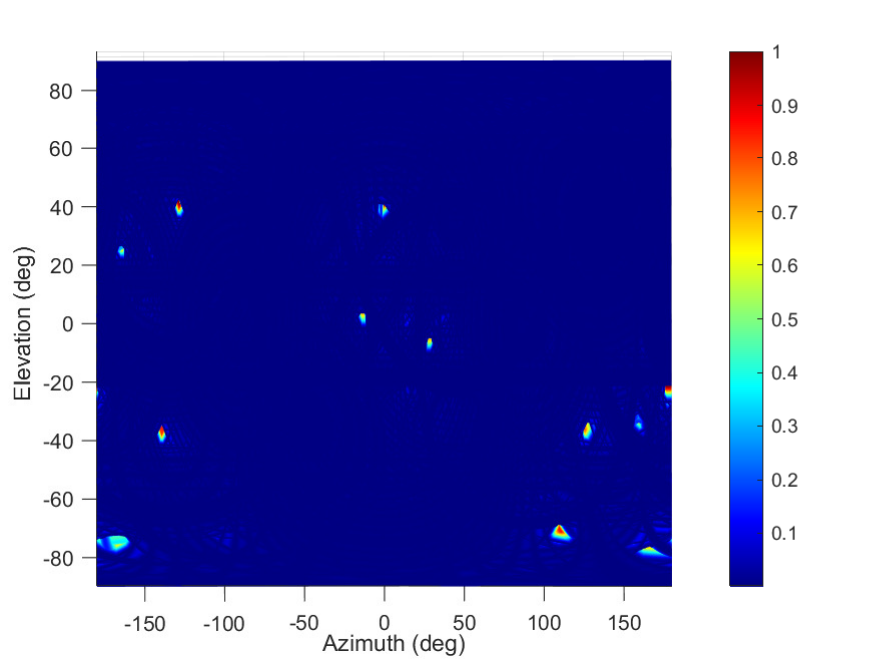}}
				\subfigure[RULA-joint with $R=7$.]{ \label{fig:RULA_joint_beam}
					\includegraphics[width=1.706in]{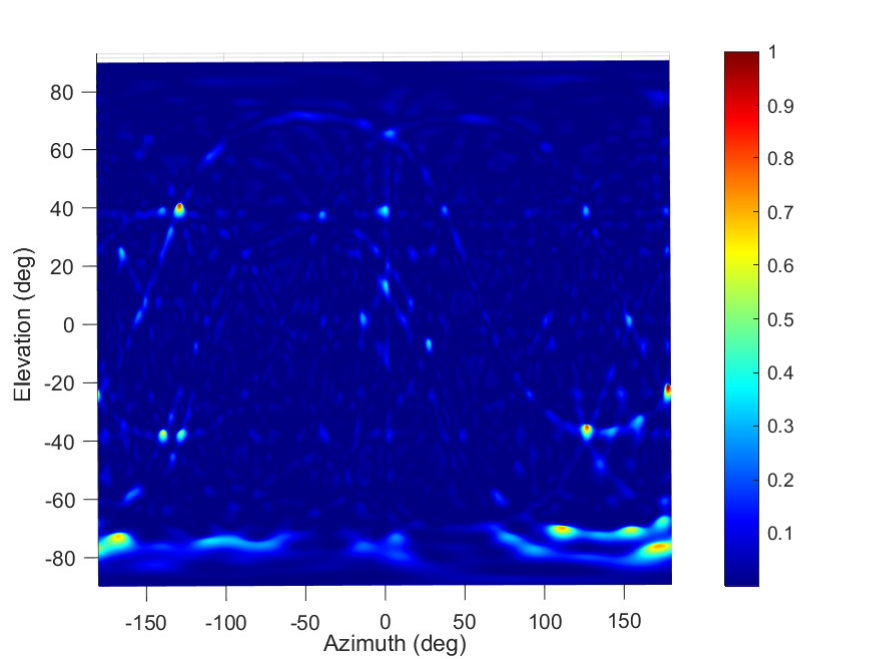}}	
				\caption{An instance with $U=49$ with SNR = 25 dB: (a). The noise-free offline beamforming with virtual 3D cubic array, which is adopted as the benchmarking ground truth. (b). The online beamforming with UPA. (c). RULA with pointwise minimum fusion. (d). RULA with joint coherent fusion.} %
				\label{fig:Synthesized_SPS}
		\end{figure*}

		\subsection{Synthesized Spatial Spectrum Quality Comparison}
		
		Fig. \ref{fig:Synthesized_SPS} compares four synthesis strategies via displaying their corresponding synthesized images on an equirectangular grid over azimuth and elevation.
		It is observed in Fig. \ref{fig:3D_cubic_beam} that the offline 3D cubic beamforming provides the cleanest and most faithful reconstruction where all 12 MPCs are sharply localized with minimal spurious peaks. 
		However, the drawback is that it requires measuring $\{r(\br{p}_n)\}$ over $U^3$ spatial positions, which is time-consuming and incurs substantial movement overhead. 
		While the widely adopted online analog receive combining with UPA successfully detects the MPCs at the front half-space with $\theta \in (0,\pi/2)$, those on the back side are inherently not observable, as shown in Fig. \ref{fig:UPA_beam}. Besides, paths arriving close to the plane of the UPA exhibit elongated peaks due to the shrinking projected aperture.
		Fig. \ref{fig:RULA_CT_beam} shows that the proposed RULA with pointwise minimum fusion (RULA-min) with $R=7$ achieves a spatial spectrum that is qualitatively close to the 3D cubic reference as all the 12 MPCs are detected across the full 3D space. 
		It is important to note that the spectrum is obtained using only $7U$ spatial samples, which enable a much more practical online synthesis pipeline while maintaining near optimal detectability. In contrast, the RULA with joint coherent fusion (RULA-joint) also recovers several MPC peaks, but exhibits prominent sidelobe and arc-shaped artifacts, indicating that the joint synthesis is more prone to sidelobe effects under sparse sampling.

		
		
		To compare the performance between different schemes quantitatively, we adopt the region-of-interest (ROI) - structural similarity index measure (SSIM) with spherical weighting \cite{hua2026rotating} and choose the image synthesized by the conventional noise-free 3D cubic approach as the reference.
		
		Specifically, let $A(\phi,\theta)$ and $B(\phi,\theta)$ denote two equirectangular spatial spectrum images sampled on a uniform grid
		$\phi\in[-\pi,\pi)$ 
		and $\theta\in[-\pi/2,\pi/2]$,
		yielding images $A,B\in\mathbb{R}^{N_\theta\times N_\phi}$.
		For each pixel $(i,j)$, the local SSIM is computed over a window $\Omega_{ij}$ (e.g., Gaussian-weighted)
		as \cite{zhao2023nerf2}
		\begin{equation}
			\mathrm{SSIM}_{ij}
			=
			\frac{\bigl(2\mu_A\mu_B + C_1\bigr)\bigl(2\sigma_{AB}+C_2\bigr)}
			{\bigl(\mu_A^2+\mu_B^2+C_1\bigr)\bigl(\sigma_A^2+\sigma_B^2+C_2\bigr)}.
		\end{equation}
		To avoid background-dominated similarity due to sparsity of MPCs, we restrict the averaging to an
		ROI defined by a joint local-maximum threshold $\tau$ within each $\Omega_{ij}$,
		\begin{equation}
			\nonumber
			M_{ij} \triangleq \max_{(p,q)\in\Omega_{ij}} \max(A(p,q),B(p,q)), \enspace
			\mathcal{M}_{ij} \triangleq \mathbf{1}\{M_{ij}>\tau\},
		\end{equation}
		where 
		$\mathcal{M}_{ij}\in\{0,1\}$ is the ROI binary mask and $\tau$ is set to be 0.4 unless otherwise stated.
		Finally, to ensure a fair comparison on the 3D sphere, we weight pixels by the spherical surface-area element. For an equirectangular parameterization, we have $d\Omega=\cos\theta\,d\phi\,d\theta$, hence we use \(W_{ij} = \cos\theta_i,\enspace \theta_i\in(-\pi/2,\pi/2),\)
		and define the ROI-SSIM score as the ROI masked, spherical-area-weighted mean:
		\begin{equation}
			\mathrm{ROI\text{-}SSIM}
			\triangleq
			\frac{\sum_{i,j} W_{ij}\,\mathcal{M}_{ij}\,\mathrm{SSIM}_{ij}}
			{\sum_{i,j} W_{ij}\,\mathcal{M}_{ij}}.
		\end{equation}
		\begin{figure}[t]
			\centering
			\includegraphics[width=2.5in]{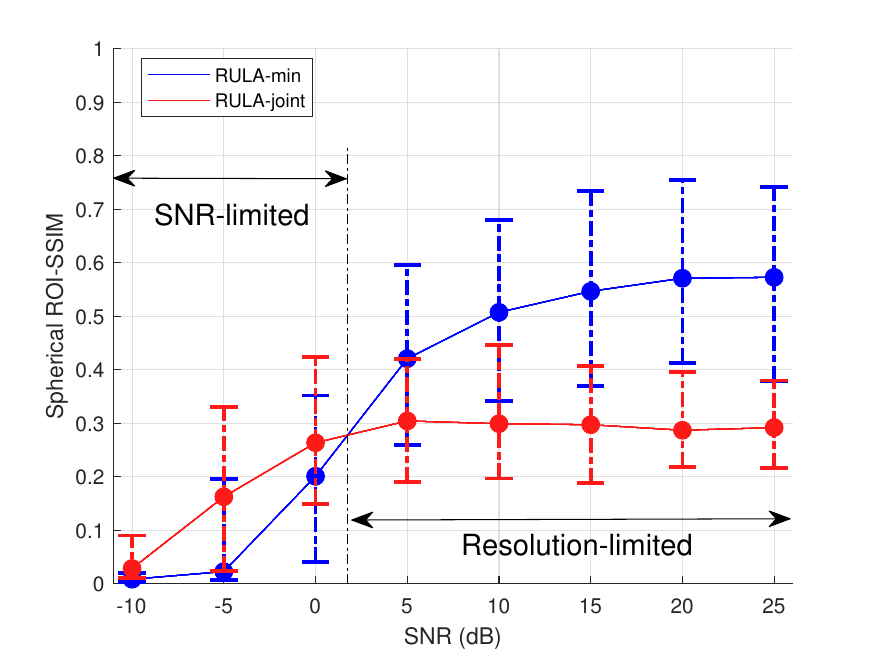}
			\caption{Spherical ROI-SSIM versus SNR.} 
			\label{fig:SSIM_vs_SNR}
		\end{figure}
		
		With the aforementioned setup with 100 MC trials, Fig. \ref{fig:SSIM_vs_SNR} shows the spherical ROI-SSIM as a function of SNR with $U = 65$. For each SNR bin, the maximum SSIM, minimum SSIM, and the mean SSIM over 100 MC trials are shown. It is observed that RULA-joint yields higher ROI-SSIM than RULA-min at low SNR as it coherently aggregates power from all $7U$ samples, which reduces the effective noise level. This makes it more robust in the SNR-limited regime, where SNR-enhancement via power aggregation from more sampling positions is more beneficial than sidelobe suppression. However, as SNR increases, RULA-min becomes consistently better and eventually saturates at a substantially higher ROI-SSIM than RULA-joint. In this region, additive noise is no longer the dominant impairment; instead, the limiting factor is the structural distortion of the synthesized spectrum, particularly sidelobe artifacts and ambiguity-induced spurious peaks. RULA-min is designed to reduce such artifacts through the minimum fusion operator, which tends to retain only directionally consistent peaks across rotations while rejecting rotation-dependent sidelobes. This region can therefore be interpreted as resolution-limited: improvements in SNR bring diminishing returns, and the dominant performance bottleneck becomes the sidelobe artifacts.
		
		\begin{figure}[t]
			\centering
			\includegraphics[width=2.5in]{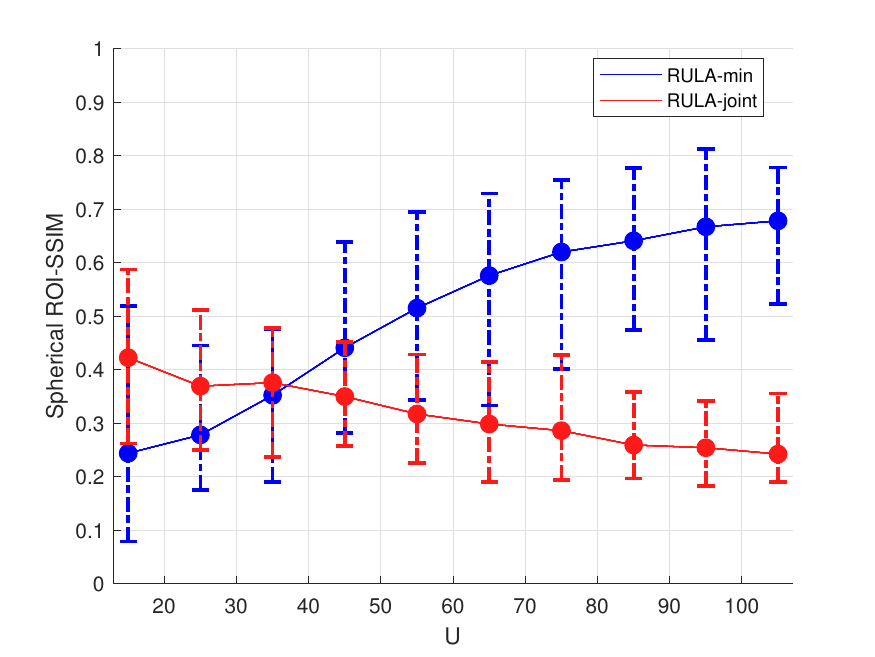}
			\caption{Spherical ROI-SSIM versus $U$.}
			\label{fig:SSIM_vs_U} 
		\end{figure}
		Fig. \ref{fig:SSIM_vs_U} shows the spherical ROI-SSIM versus $U$ with SNR  $=$ 20 dB. As $U$ increases, the performance of RULA-min improves monotonically and eventually saturates at a high level. This trend is expected for two reasons. First, a larger ULA provides a larger effective aperture, so the synthesized spectrum benefits from higher peak signal-to-noise ratio (PSNR) due to coherent integration across more spatial samples. Second, the angular resolution improves with $U$ (narrower mainlobe), which yields sharper, better-localized peaks and reduces peak overlap. In contrast, RULA-joint exhibits a decreasing ROI-SSIM as $U$ grows since when $U$ is small, the reference spectrum is relatively broad, and the joint aggregation across all $7U$ samples provides a strong integration gain that stabilizes peak detection. However, when $U$ increases, the reference spectrum becomes progressively sharper due to the enlarged 3D aperture, while RULA-joint remains limited by its joint combination artifacts and coarser peaks, as already shown in Fig. \ref{fig:RULA_joint_beam}. As a result, the mismatch between the RULA-joint spectrum and the reference increases with $U$, leading to a gradual reduction in the spherical ROI-SSIM.
		
		\section{Conclusion}
		
		This paper proposed a rotating-ULA-enabled framework for online full-space SPS synthesis. By rotating a ULA around its center, the receiver performs analog receive combining at multiple orientations and obtains a set of partial observations using a single RF chain. Under this sampling geometry, two rules including both pointwise minimum fusion and joint coherent fusion were developed for spectrum synthesis.
		A principled environment-agnostic orientation design method for selecting the ULA axes was also developed. The primary criterion maximizes the expected worst-case projected separation between random MPC pairs, for which the globally optimal relaxed design is shown to be isotropic in the sense that the summed outer products of all orientation axes form a scaled identity matrix. This condition can be exactly realized by finite orientation sets for all \(R\ge 3\). A secondary criterion then reduces orientation redundancy by minimizing the worst-case projective correlation under the isotropic-matrix constraint, and a multi-start smooth minimax algorithm was used to obtain concrete orientation realizations.
		Numerical results verified the large-\(K\) scaling law and showed that the optimized orientation sets are well distributed over the sphere. For SPS synthesis, the proposed RULA-min method reconstructed full-space 3D spectra close to the dense 3D cubic reference in the high-SNR regime while using only \(O(U)\) spatial samples instead of \(U^3\). The comparison with RULA-joint further showed that coherent aggregation is beneficial in the low-SNR regime, whereas the pointwise minimum fusion is more effective in the resolution-limited regime because it suppresses sidelobe and ambiguity artifacts. These results demonstrate that the proposed rotating-ULA CT framework provides a low-complexity and practical approach for efficient full-space wireless channel characterization.
		\appendices
		\section{Proof of Proposition~\ref{prop:isotropic_G_optimal}}\label{app:proof_isotropic_G}
		For a fixed realization of \(\{\br s_k\}\), the inner minimum in \(F(\br G)\) is the pointwise minimum of linear functions of \(\br G\), and is hence concave in \(\br G\). Taking expectation preserves concavity, so \(F(\br G)\) is concave. The positive-semidefinite constraint and the affine trace constraint in \eqref{eq:relaxed_G_design} are also convex.

		The uniform spherical distribution is rotationally invariant. Thus, for any rotation matrix \(\br U\in SO(3)\), where \(SO(3)\) is the group of 3D rotations, we have
		\begin{align}\label{eq:rot_invar_F}
			F(\br G)=F(\br U\br G\br U^T).
		\end{align}
		The constraints are also invariant since \(\br G\succeq\br 0\) implies \(\br U\br G\br U^T\succeq\br 0\), and
		\(\operatorname{tr}\{\br U\br G\br U^T\}=\operatorname{tr}\{\br G\}\). Let \(P\) denote the normalized Haar probability measure on \(SO(3)\), for which \(\int_{SO(3)}\mathrm dP(\br U)=1\) \cite{mezzadri2007generate}. Averaging all rotated copies of any feasible \(\br G\) over \(SO(3)\) gives
\(\bar{\br G} = \int_{SO(3)}\br U\br G\br U^T \mathrm dP(\br U) .\)
		By Jensen's inequality and the concavity of \(F\), applied to the random rotated matrix \(\br U\br G\br U^T\), we have
		\begin{align}
			F(\bar{\br G}) &=
			F\left(
			\mathbb E_{\br U}
			\left[
			\br U\br G\br U^T
			\right]
			\right) 
			\ge
			\mathbb E_{\br U}
			\left[
			F(\br U\br G\br U^T)
			\right] \nonumber\\
			&=
			\mathbb E_{\br U}
			\left[
			F(\br G)
			\right]
			=
			F(\br G),
		\end{align}
		where \(\mathbb E_{\br U}[\cdot]\) is with respect to \(P\), and the second equality follows from  \eqref{eq:rot_invar_F}.
		Hence, the rotational average of any feasible design is at least as good as the original one.

		It remains to identify the structure of \(\bar{\br G}\). For any fixed rotation \(\br V\in SO(3)\),
		\begin{align}\label{eq:structure_G_bar}
			\br V\bar{\br G}\br V^T
			&=
			\int_{SO(3)}
			(\br V\br U)\br G(\br V\br U)^T
			\mathrm dP(\br U) .
		\end{align}
		Since \(\br V\br U\) is also an element of \(SO(3)\), and the normalized Haar measure is invariant to this left multiplication, (\ref{eq:structure_G_bar}) is simply a relabeling of the same rotational average. Therefore,
		\begin{align}
			\nonumber
			\br V\bar{\br G}\br V^T
			=
			\int_{SO(3)}
			\br U\br G\br U^T
			\mathrm dP(\br U)
			=
			\bar{\br G},
			\quad \forall \br V\in SO(3).
		\end{align}
		Thus, \(\bar{\br G}\) is invariant under every rotation and has no preferred spatial direction. The only \(3\times 3\) symmetric matrices with this invariance are scalar multiples of the identity. Hence \(\bar{\br G}=c\br I_3\). Since trace is preserved by rotations,
		\begin{align}
			\operatorname{tr}\{\bar{\br G}\}
			=
			\operatorname{tr}\{\br G\}
			=
			R,
		\end{align}
		which gives \(c=R/3\). Therefore, we have \(\bar{\br G}=\frac{R}{3}\br I_3\), and since \(F(\bar{\br G})\ge F(\br G)\) for any feasible \(\br G\), \(\br G^\star=\frac{R}{3}\br I_3\) is globally optimal for \eqref{eq:relaxed_G_design}.
		Notice that this argument establishes global optimality but not uniqueness. The non-uniqueness is most transparent in the special case \(K=2\). 
		In this case, the inner minimum contains only one pair, and thus we have
		\begin{align}
			F(\br G)
			&=
			\mathbb E_{\br s_1,\br s_2}
			\left[
			(\br s_1-\br s_2)^T
			\br G
			(\br s_1-\br s_2)
			\right] \nonumber\\
			&=
			\operatorname{tr}
			\left\{
			\br G
			\mathbb E
			\left[
			(\br s_1-\br s_2)(\br s_1-\br s_2)^T
			\right]
			\right\}.
		\end{align}
		Since \(\br s_1\), \(\br s_2\) are independent and uniformly distributed on \(S^2\),
		\begin{align}
			\mathbb E
			\left[
			(\br s_1-\br s_2)(\br s_1-\br s_2)^T
			\right]
			=
			\frac{2}{3}\br I_3.
		\end{align}
		Therefore,
		\begin{align}
			F(\br G)
			=
			\frac{2}{3}\operatorname{tr}\{\br G\}
			=
			\frac{2R}{3},
		\end{align}
		for every feasible \(\br G\). Thus, all feasible PSD matrices satisfying \(\operatorname{tr}\{\br G\}=R\) are globally optimal when \(K=2\). Intuitively, with only one random MPC pair, 
		there is no competition among multiple pairwise separations.
		The isotropic choice remains a balanced global optimum, but it is not uniquely selected by this objective until multiple possible MPC pairs are present, 
		i.e., \(K>2\).

		\section{Proof of Proposition~\ref{prop:tight_frame_existence}}\label{app:proof_tight_frame_existence}
		We prove the result by construction. Let \(R\ge 3\), and define
		\begin{align}
			\alpha_i=\frac{2\pi(i-1)}{R},
			\enspace i=1,\ldots,R.
		\end{align}
		By choosing
		\(
			\br q_i
			=
			\left[
			\sqrt{\frac{2}{3}}\cos\alpha_i,
			\sqrt{\frac{2}{3}}\sin\alpha_i,
			\frac{1}{\sqrt{3}}
			\right]^T,
			\enspace i=1,\ldots,R, 
		\)
		it is easy to see that \(\br q_i\) has unit norm, since
		\begin{align}
			\|\br q_i\|^2
			=
			\frac{2}{3}\cos^2\alpha_i
			+
			\frac{2}{3}\sin^2\alpha_i
			+
			\frac{1}{3}
			=
			1.
		\end{align}
		Because the angles \(\{\alpha_i\}\) are uniformly spaced on the unit circle and \(R\ge 3\), we have \(
			\sum_{i=1}^{R}\cos\alpha_i
			=
			\sum_{i=1}^{R}\sin\alpha_i
			=
			0,\)
		\(
			\sum_{i=1}^{R}\cos^2\alpha_i
			=
			\sum_{i=1}^{R}\sin^2\alpha_i
			=
			\frac{R}{2}, \) and \(
			\sum_{i=1}^{R}\sin\alpha_i\cos\alpha_i=0.\)
		Using these identities, the outer-product sum becomes
		\begin{align}
			\sum_{i=1}^{R}\br q_i\br q_i^T 
			=
			\begin{bmatrix}
				\frac{R}{3} & 0 & 0\\
				0 & \frac{R}{3} & 0\\
				0 & 0 & \frac{R}{3}
			\end{bmatrix} 	=	\frac{R}{3}\br I_3.
		\end{align}
		Thus, the isotropic Gram matrix in \eqref{eq:isotropic_G_optimal} admits a decomposition into \(R\) unit-norm orientations for every \(R\ge 3\).

		\section{Proof of Proposition~\ref{prop:isotropic_objective_value}}\label{app:proof_isotropic_objective_value}
		Substituting \(\br G^\star=\frac{R}{3}\br I_3\) into \(F(\br G)\) gives
		\begin{align}\label{eq:first_sim}
			F(\br G^\star)
			&=
			\frac{R}{3}
			\mathbb E
			\left[
			\min_{k\ne \ell}
			\|\br s_k-\br s_\ell\|^2
			\right]
			=
			\frac{R}{3}\mathbb E[D_{\min}].
		\end{align}
		Since each MPC direction has unit norm, let \(\gamma_{k\ell}\) denote the angle between \(\br s_k\) and \(\br s_\ell\), then
		\begin{align}
			\|\br s_k-\br s_\ell\|^2
			=
			\|\br s_k\|^2+\|\br s_\ell\|^2
			-2\br s_k^T\br s_\ell
			=
			2-2\cos\gamma_{k\ell}.
		\end{align}

		For \(K=2\), by rotational invariance, we may fix \(\br s_1=[0,0,1]^T\) without loss of generality. Then \(\br s_1^T\br s_2\) is the \(z\)-coordinate \(Z\) of a uniformly distributed point on \(S^2\). The distribution on \(S^2\) is uniform with respect to surface area, so the probability of \(Z\) falling in an interval is proportional to the surface area of the corresponding horizontal spherical band. Equivalently, for \(-1\le z\le 1\), the event \(Z\le z\) corresponds to the spherical region below height \(z\), whose surface area is \(2\pi(z+1)\). Since the total surface area of \(S^2\) is \(4\pi\),
		\begin{align}
			F_Z(z)
			=
			\Pr(Z\le z)
			=
			\frac{2\pi(z+1)}{4\pi}
			=
			\frac{z+1}{2}.
		\end{align}
		Thus \(Z\sim\operatorname{Unif}[-1,1]\), and
		\begin{align}
			\mathbb E[D_{\min}]
			=
			\mathbb E[2-2\br s_1^T\br s_2]
			=
			2.
		\end{align}
		This gives \(F(\br G^\star)=2R/3\) for \(K=2\).

		It remains to examine the case when $K$ is sufficiently large. Let \(N_p=K(K-1)/2\). For a fixed \(x\ge 0\), define the close-pair count as
		\begin{align}
			\nonumber
			N_K(x)
			=
			\sum_{k<\ell}
			I_{k\ell}(x),
			\quad
			I_{k\ell}(x)
			=
			\mathbf 1
			\left\{
			\|\br s_k-\br s_\ell\|^2\le \frac{x}{N_p}
			\right\}.
		\end{align}
		The event \(N_pD_{\min}>x\) is equivalent to \(N_K(x)=0\). For any fixed pair \((k,\ell)\), let \(Z_{k\ell}=\br s_k^T\br s_\ell\). Since \(Z_{k\ell}\sim\operatorname{Unif}[-1,1]\) and \(\|\br s_k-\br s_\ell\|^2=2-2Z_{k\ell}\), we have
		\begin{align}
			\mathbb E[I_{k\ell}(x)]
			=
			\Pr\left(
			\|\br s_k-\br s_\ell\|^2\le \frac{x}{N_p}
			\right)
			=
			\frac{x}{4N_p}.
		\end{align}
		Hence we have \(\mathbb E[N_K(x)]=x/4\).

		The indicators \(\{I_{k\ell}(x)\}\) are rare. They are independent unless the two pairs share one index, and each pair shares an index with only \(O(K)\) other pairs. For example, conditioning on \(\br s_k\), the events \(I_{k\ell}(x)=1\) and \(I_{kj}(x)=1\) have probability
		\begin{align}
			\Pr(I_{k\ell}(x)=1,I_{kj}(x)=1)
			=
			\left(\frac{x}{4N_p}\right)^2 .
		\end{align}
		There are only \(O(K^3)\) such dependent pairwise combinations, but their total contribution is shown to vanish as \(K\to\infty\), i.e.,
		\begin{align}
			O(K^3)
			\left(\frac{x}{4N_p}\right)^2
			=
			O(K^{-1})
			\rightarrow 0, \enspace K \to \infty.
		\end{align}
		Thus, the close-pair count is asymptotically Poisson with mean \(x/4\) according to the Poisson approximation for rare locally dependent events \cite{arratia1989two}, i.e.,
		\begin{align}
			N_K(x)
			\Rightarrow
			\operatorname{Poisson}\left(\frac{x}{4}\right).
		\end{align}
		Consequently,
		\begin{align}
			\Pr(N_pD_{\min}>x)
			=
			\Pr(N_K(x)=0)
			\rightarrow
			\exp\left(-\frac{x}{4}\right),
		\end{align}
		which proves \(N_pD_{\min}\Rightarrow \operatorname{Exp}(1/4)\). Therefore, for large \(K\), 
		\begin{align}
			\mathbb E[D_{\min}]
			\approx
			\frac{4}{N_p}
			=
			\frac{8}{K(K-1)}.
		\end{align}
		Substituting this into \eqref{eq:first_sim} yields \eqref{eq:isotropic_F_largeK}.

		\bibliographystyle{IEEEtran}
		
		\bibliography{refsv3}

	\end{document}